\documentclass{article}

\usepackage{microtype}
\usepackage{graphicx}
\usepackage{subcaption}
\usepackage{booktabs} 
\usepackage{enumitem} 

\usepackage{hyperref}

\usepackage[preprint]{icml2026}

\usepackage[utf8]{inputenc}
\usepackage[T1]{fontenc}

\usepackage{url}
\usepackage{booktabs}
\usepackage{amsfonts}
\usepackage{nicefrac}
\usepackage{microtype}
\usepackage{xcolor}

\usepackage{amsmath,bm}
\usepackage{amsthm}
\usepackage{amssymb}
\usepackage{mathtools}
\usepackage{centernot}
\usepackage{algorithm}
\usepackage{comment}
\usepackage{tikz}
\usepackage{amsmath, amssymb}
\usetikzlibrary{positioning, shapes, arrows.meta, calc, backgrounds, fit, decorations.pathreplacing}
\usetikzlibrary{shadows}

\definecolor{nodeblue}{RGB}{66, 133, 244}
\definecolor{rootred}{RGB}{219, 68, 55}
\definecolor{linkgray}{RGB}{150, 150, 150}
\definecolor{highlightbg}{RGB}{230, 240, 255}
\definecolor{cloudbg}{RGB}{245, 245, 245}

\newcommand{\E}{\operatorname{\mathbb E}}
\newcommand{\Erdos}{Erd\H{o}s}
\newcommand{\Renyi}{R\'enyi}

\newtheorem{lemma}{Lemma}
\newtheorem{remark}{Remark}
\newtheorem{corollary}{Corollary}
\newtheorem{theorem}{Theorem}
\newtheorem{proposition}{Proposition}
\newtheorem{definition}{Definition}
\newtheorem{assumption}{Assumption}

\usepackage[capitalize,noabbrev]{cleveref}

\usepackage[textsize=tiny]{todonotes}

\icmltitlerunning{Mean-Field Control on Sparse Graphs: From Local Limits to GNNs via Neighborhood Distributions}

\begin{document}

\twocolumn[
  \icmltitle{Mean-Field Control on Sparse Graphs: From Local Limits to GNNs via Neighborhood Distributions}



    \icmlsetsymbol{equal}{*}
    
    \begin{icmlauthorlist}
    \icmlauthor{Tobias Schmidt}{x}
    \icmlauthor{Kai Cui}{y}
    \end{icmlauthorlist}
    
    \icmlaffiliation{x}{Department of Mathematics, TU Darmstadt, Darmstadt, Germany}
    \icmlaffiliation{y}{Department of Electrical Engineering, TU Darmstadt, Darmstadt, Germany}
    
    \icmlcorrespondingauthor{Tobias Schmidt}{tobias.schmidt@tu-darmstadt.de}
    \icmlcorrespondingauthor{Kai Cui}{kai.cui@tu-darmstadt.de}
    
    \icmlkeywords{Mean Field Control, Graph Neural Networks, Sparse Graphs, Local Weak Limits, Actor Critic Algorithms}
    
  \vskip 0.3in
]



\printAffiliationsAndNotice{}  

\begin{abstract}
Mean-field control (MFC) offers a scalable solution to the curse of dimensionality in multi-agent systems but traditionally hinges on the restrictive assumption of exchangeability via dense, all-to-all interactions. In this work, we bridge the gap to real-world network structures by proposing a rigorous framework for MFC on large sparse graphs. We redefine the system state as a probability measure over \textit{decorated rooted neighborhoods}, effectively capturing local heterogeneity. Our central contribution is a theoretical foundation for scalable reinforcement learning in this setting. We prove \textit{horizon-dependent locality}: for finite-horizon problems, an agent's optimal policy at time $t$ depends strictly on its $(T-t)$-hop neighborhood. This result renders the infinite-dimensional control problem tractable and underpins a novel Dynamic Programming Principle (DPP) on the lifted space of neighborhood distributions. Furthermore, we formally and experimentally justify the use of Graph Neural Networks (GNNs) for actor-critic algorithms in this context. Our framework naturally recovers classical MFC as a degenerate case while enabling efficient, theoretically grounded control on complex sparse topologies.
\end{abstract}

\section{Introduction}
The theory of mean-field games and control \cite{andersson2011maximum, bensoussan2013mean} has emerged as a cornerstone for the study of large-scale multi-agent systems \cite{lasry2007mean, huang2006large, carmona2018probabilistic, lauriere2022learning}. By considering the limit as the number of agents $N \to \infty$, these frameworks circumvent the curse of dimensionality inherent in multi-agent reinforcement learning (MARL) problems \citep{canese2021multi, gronauer2022multi}. The core assumption is that agents are exchangeable and interact with the collective through a tractable aggregate statistic -- the empirical distribution of states, or "mean field" (MF).

A critical implicit assumption is that the interaction graph is the complete graph. While all-to-all interaction is reasonable for systems with long-range, homogeneous interactions, it fails to capture a vast array of real-world systems with local interactions. Examples include social networks \cite{jackson2008social}, power grids \cite{dorfler2013synchronization}, biological swarms \cite{vicsek1995novel}, and robotic networks \cite{bullo2009distributed}.
The challenge of sparse graphs is profound: The local environment of each agent is distinct and stochastic, breaking symmetry that underpins classical MF theory. Recent work has begun to tackle this challenge by studying MF learning on random graphs \cite{caines2021graphon, cui2022learning, fabian2023learning, zhang2024learning, fabian2024learning,  zhang2024learning2, fabian2025learning}, showing that a MF limit can still be obtained. However, the limiting MF is no longer a simple state distribution but is related to the structure of the graph itself.

This paper builds on these insights to construct a rigorous framework for reinforcement learning (RL) on sparse graphs. The central idea is to redefine the state of the system. Instead of a measure $\mu \in \mathcal{P}(\mathcal{X})$ on the agent's state space $\mathcal{X}$, the state is a measure $\boldsymbol{\mu} \in \mathcal{P}(\mathcal{G}_*^{\mathcal{X}})$ on the space of \textit{decorated rooted graphs}. This object, $\boldsymbol{\mu}$, describes the probability of observing any given local neighborhood structure, with each node in that neighborhood "decorated" by its state. For a finite $T$-horizon problem, we obtain and apply a crucial locality result: an agent's optimal policy at time $t$ depends only on the distribution of its $(T-t)$-hop neighborhood, providing a direct link between the problem horizon and the required locality of information. As a result, we have a Markov decision process (MDP) as a more tractable single-agent RL problem, and are able to design appropriate RL algorithms using GNNs with strong theoretical justification.

\subsection{Related Work}
A large body of work studies how to compute equilibria and optimal controls in MFG/MFC using RL and deep learning \cite{perrin2020fictitious, gu2021mean, lauriere2022learning}.
In the finite-state/discrete-time setting, foundational RL approaches for MFGs include Q-learning-style algorithms that couple best responses with the MF fixed point \cite{guo2019learning}, and scalable deep RL variants that target large state spaces \cite{lauriere2022scalable}.
Complementary learning paradigms include fictitious play and extensions (e.g., common noise), with convergence analyses and practical implementations \cite{perrin2020fictitious, lauriere2022scalable} and generalization viewpoints \cite{perrin2021generalization}.
For mean-field \emph{control} (cooperative) rather than games, various policy-gradient methods have also been developed \cite{gu2020q, gu2023dynamic, cui2024learning2, guan2024zero}.

A second related strand addresses \emph{many-agent RL} by replacing interactions with a mean-field (average-neighbor) surrogate, yielding mean-field Q-learning / actor--critic algorithms with strong empirical performance in large populations \cite{yang2018mean,gu2021mean}.
These methods typically assume that the relevant aggregate seen by an agent is an empirical average (often over neighbors), but they do not build a \emph{graph-local} limiting distribution over rooted neighborhoods, and use a separate set of assumptions.

\noindent\textbf{Large graph theory and mean fields.}
To incorporate network heterogeneity, several works adopt \emph{dense} graph limits via graphons, leading to graphon MFG/MFC and learning algorithms when interactions vary by latent graphon index \cite{cui2022learning,hu2023graphon,zhang2024learning,zhang2024learning2}.
More recently, various graphon variants have been proposed to better capture sparsity (e.g.\ $L^p$-graphons) beyond dense graphons \cite{fabian2023learning, fabian2024learning}.

On the probability/theory side, interacting processes on sparse graphs have a limiting description governed by local weak (Benjamini--Schramm) convergence, where the natural object is the law of a rooted neighborhood process; continuity of the dynamics with respect to local weak convergence is established in general frameworks \cite{lacker2023local} and successfully to MFC via heuristic algorithms \cite{fabian2025learning}.
Our contribution differs in emphasis and structure: we turn the local-weak-limit viewpoint into a \emph{dynamic programming} framework for control by taking the system state to be a distribution over decorated rooted neighborhoods, proving a \emph{horizon-dependent locality} principle that identifies the minimal radius required at each time step, and theoretically justifies the use of Graph Neural Networks as local policy approximators with approximate optimality. Our results show the limits of the heuristic but successful approach in \cite{fabian2025learning}, contributing a rigorous solution and theoretical basis for algorithms.

\noindent\textbf{MARL on graphs}
Finally, our algorithmic blueprint connects to extensive MARL literature on graph-structured environments and GNN-based policies/value functions \cite{Jiang2020Graph,sunehag2018value,rashid2020monotonic}, but these works typically proceed from a finite-$N$ training objective rather than from an explicit sparse-graph mean-field limit. Our work provides separate motivation for architecture choices such as GNN especially in large sparse graphs where neighborhoods are of interest up to isomorphism.

In contrast to classical MFC and learning, where the MF is a distribution over agent states and interaction structure is implicit (typically all-to-all), our framework replaces the mean field by a distribution over decorated rooted neighborhoods arising from the local weak limit of sparse graphs as shown in Fig.~\ref{fig:overview1}. This shift preserves classical MFC as a special case while enabling a dynamic programming formulation that remains valid when interactions are local and the population is no longer exchangeable.

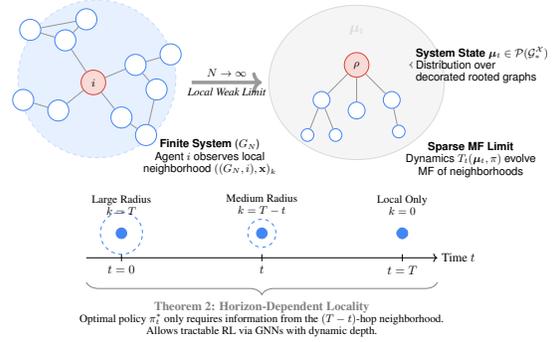
\begin{figure}[t]
  \centering
  \resizebox{0.88\linewidth}{!}{
    \begin{tikzpicture}[
        node distance=1.5cm,
        agent/.style={circle, draw=nodeblue, fill=white, thick, minimum size=0.6cm, inner sep=0pt, font=\small},
        root/.style={circle, draw=rootred, fill=rootred!20, thick, minimum size=0.7cm, inner sep=0pt, font=\bfseries\small},
        edge/.style={draw=linkgray, thick},
        arrow/.style={->, >=Stealth, thick, shorten >=2pt, shorten <=2pt},
        labeltext/.style={font=\sffamily\small, align=center},
        box/.style={rounded corners, draw=gray!50, fill=white, dashed}
    ]
    
        
        \begin{scope}[local bounding box=finiteGraph]
            \node[root] (r) at (0,0) {$i$};
            
            \node[agent] (n1) at (1.2, 0.5) {};
            \node[agent] (n2) at (-0.8, 1.2) {};
            \node[agent] (n3) at (-1.0, -0.8) {};
            \node[agent] (n4) at (0.8, -1.0) {};
            
            \node[agent] (nn1) at (2.2, 0.8) {};
            \node[agent] (nn2) at (1.8, -0.2) {};
            \node[agent] (nn3) at (-1.8, 1.5) {};
            \node[agent] (nn4) at (-0.2, 1.8) {};
            \node[agent] (nn5) at (-2.0, -0.5) {};
            \node[agent] (nn6) at (1.5, -1.5) {};
            
            \draw[edge] (r) -- (n1);
            \draw[edge] (r) -- (n2);
            \draw[edge] (r) -- (n3);
            \draw[edge] (r) -- (n4);
            
            \draw[edge] (n1) -- (nn1);
            \draw[edge] (n1) -- (nn2);
            \draw[edge] (n2) -- (nn3);
            \draw[edge] (n2) -- (nn4);
            \draw[edge] (n3) -- (nn5);
            \draw[edge] (n4) -- (nn6);
            
            \draw[edge] (nn2) -- (nn6);
            \draw[edge] (nn3) -- (nn4);
        \end{scope}
    
        \begin{scope}[on background layer]
            \node[circle, fill=highlightbg, draw=nodeblue!30, dashed, fit=(r) (n1) (n2) (n3) (n4), inner sep=5pt] (neighborhood) {};
        \end{scope}
    
        \node[below right=-0.3cm and -1.2cm of finiteGraph, labeltext] (lbl1) {
            \textbf{Finite System} ($G_N$)\\
            Agent $i$ observes local\\neighborhood $((G_N, i), \mathbf{x})_k$
        };
    
        
        \node[right=1.5cm of r, font=\large] (arrow_start) {};
        \draw[->, line width=1mm, draw=gray!80] (2.8,0) -- (4.8,0) node[midway, above, font=\bfseries\small] {$N \to \infty$} node[midway, below, font=\itshape\footnotesize] {Local Weak Limit};
    
        
        \begin{scope}[shift={(7.5,0)}, local bounding box=limitSpace]
            
            \draw[fill=cloudbg, draw=gray!30] (0,0) ellipse (2.5cm and 2.2cm);
            \node[font=\bfseries\large, gray!30] at (0,1.5) {$\boldsymbol{\mu}_t$};
            
            \node[root] (lim_r) at (0,0.5) {$\rho$};
            \node[agent, scale=0.8] (lim_n1) at (-1, -0.5) {};
            \node[agent, scale=0.8] (lim_n2) at (0, -0.8) {};
            \node[agent, scale=0.8] (lim_n3) at (1, -0.5) {};
            
            \node[agent, scale=0.6] (lim_nn1) at (-1.4, -1.5) {};
            \node[agent, scale=0.6] (lim_nn2) at (-0.6, -1.5) {};
            \node[agent, scale=0.6] (lim_nn3) at (1.2, -1.4) {};
    
            \draw[edge] (lim_r) -- (lim_n1);
            \draw[edge] (lim_r) -- (lim_n2);
            \draw[edge] (lim_r) -- (lim_n3);
            \draw[edge] (lim_n1) -- (lim_nn1);
            \draw[edge] (lim_n1) -- (lim_nn2);
            \draw[edge] (lim_n3) -- (lim_nn3);
            
            \node[right=1.2cm of lim_r, labeltext, align=left, anchor=west] (state_def) {
                \textbf{System State} $\boldsymbol{\mu}_t \in \mathcal{P}(\mathcal{G}_*^{\mathcal{X}})$\\
                Distribution over\\decorated rooted graphs
            };
            \draw[->, gray, dashed] (state_def.west) -- (1.5, 0.5);
            
        \end{scope}
    
        \node[below right=-0.63cm and -4.3cm of limitSpace, labeltext] (lbl2) {
            \textbf{Sparse MF Limit}\\
            Dynamics $T_t(\boldsymbol{\mu}_t, \pi)$ evolve\\MF of neighborhoods
        };
    
        
        \begin{scope}[shift={(0.8,-0.5)}]
            
        \node[below=2.5cm of r, anchor=center] (timeline_start) {};
        
        \draw[->, thick] (-1, -4.5) -- (9, -4.5) node[right] {Time $t$};
        
        \draw[thick] (0, -4.4) -- (0, -4.6) node[below] {$t=0$};
        \node[agent, scale=0.5, fill=nodeblue] at (0, -3.8) {};
        \draw[dashed, nodeblue] (0, -3.8) circle (0.6cm);
        \node[above, font=\small, align=center] at (0, -3.4) {Large Radius\\$k=T$};
    
        \draw[thick] (4, -4.4) -- (4, -4.6) node[below] {$t$};
        \node[agent, scale=0.5, fill=nodeblue] at (4, -3.8) {};
        \draw[dashed, nodeblue] (4, -3.8) circle (0.4cm);
        \node[above, font=\small, align=center] at (4, -3.4) {Medium Radius\\$k=T-t$};
    
        \draw[thick] (8, -4.4) -- (8, -4.6) node[below] {$t=T$};
        \node[agent, scale=0.5, fill=nodeblue] at (8, -3.8) {};
        \draw[dashed, nodeblue] (8, -3.8) circle (0.15cm);
        \node[above, font=\small, align=center] at (8, -3.4) {Local Only\\$k=0$};
    
        \draw[decorate, decoration={brace, amplitude=10pt, mirror}, thick, gray] (-1, -5.2) -- (9, -5.2) node[midway, below=12pt, font=\bfseries] {Theorem 2: Horizon-Dependent Locality};
        \node[below=22pt, align=center, font=\small] at (4, -5.2) {Optimal policy $\pi_t^*$ only requires information from the $(T-t)$-hop neighborhood.\\Allows tractable RL via GNNs with dynamic depth.};

        \end{scope}
    
    \end{tikzpicture}
  }
  \caption{\textbf{Overview of the Sparse Mean-Field Control Framework.} (Top Left) The physical $N$-agent system on a sparse graph where agents interact locally. (Top Right) The rigorous mean-field limit where the state is redefined as a probability measure $\boldsymbol{\mu}_t$ over the space of decorated rooted graphs. (Bottom) Our Horizon-Dependent Locality result (Theorem \ref{thm:locality}) establishes the $(T-t)$-hop neighborhood at time $t$ as sufficient statistic for optimal control.}
  \label{fig:overview1}
\end{figure}

\subsection{Core Contributions}
Our main contributions are: (i) A novel mathematical framework with optimality characterized to horizon-dependent locality, using the distribution over decorated neighborhoods and deriving the corresponding state transition dynamics; (ii) a DPP and MDP, establishing a Bellman equation and proving the existence of an optimal local policy under standard continuity and compactness; (iii) an algorithmic blueprint where we propose practical RL algorithms on sparse large graphs, that are theoretically justified as the RL solution to a well-defined MDP; and (iv) theoretical guarantees that show approximate optimality of mean-field solutions through a "propagation of chaos" type result, a policy gradient theorem for sparse graphs, as well as a rigorous justification for GNNs in RL on large sparse graphs. This work rigorously bridges the gap between classical mean-field learning algorithms and the reality of highly sparse systems, providing both a rigorous foundation and a practical pathway for designing decentralized control strategies on large networks.

\section{Background and Preliminaries}
In this section, we expand upon classical MFC and our proposed framework. We begin by explaining how classical MFC and our approach fit together.

\subsection{From Classical MFC to Graph-Lifted MFC}

In a standard $N$-agent MFC problem, each agent $i \in \{1, \dots, N\}$ has a state $x_t^i \in \mathcal{X}$ and chooses an action $u_t^i \in \mathcal{U}$. The dynamics and reward for agent $i$ depend on its own state and action, and on the empirical distribution of all agent states, $\mu_t^N = \frac{1}{N}\sum_{j=1}^N \delta_{x_t^j}$. In the limit $N \to \infty$, the control problem simplifies to solving a Bellman equation on the space $\mathcal{P}(\mathcal{X})$ for the value function $V(t, \mu)$; see also \cite{gu2023dynamic} for a more thorough introduction. This formulation relies on the all-to-all interaction structure, which is a strong symmetry assumption: agents are exchangeable and interact only through the empirical distribution of states $\mu_t \in \mathcal P(X)$. While this assumption is natural when interactions are dense or effectively all-to-all, it is unreasonable in systems where interactions are mediated by a sparse graph.

From a modeling perspective, the failure is structural rather than technical. In sparse networks, an agent does not interact with a representative sample of the population, but with a small, random neighborhood whose composition varies across agents. As a result, the empirical state distribution $\mu_t$ is no longer a closed sufficient statistic for the dynamics or rewards, even in the limit $N \to \infty$.

Our framework can be understood as a minimal extension of classical MFC by lifting the mean field from states to \emph{local environments}. Concretely, instead of tracking only the marginal distribution of agent states, we consider the distribution of \emph{decorated rooted neighborhoods},
\[
\boldsymbol\mu_t \;\in\; \mathcal P(\mathcal G^\mathcal X_\ast),
\]
where each element describes both the local interaction structure and the states of agents within that neighborhood.
This lift has three key consequences.

\noindent\textbf{(i) Closure under Sparse Dynamics.}
Local transition kernels and rewards depend on bounded neighborhoods by assumption. By taking the system state to be a distribution over such neighborhoods, the limiting dynamics become deterministic/Markovian at the level of $\boldsymbol\mu_t$, exactly as in classical MFC, but now without requiring dense interactions.

\noindent\textbf{(ii) Recovery of Classical MFC as a Special Case.}
When the interaction graph is complete, or when rewards and dynamics depend only on the root state, the neighborhood distribution $\boldsymbol \mu_t$ collapses to the marginal state distribution in $\mathcal P(\mathcal X)$. In this sense, classical mean-field control is a degenerate case of the graph-lifted formulation, corresponding to zero-radius neighborhoods.

\noindent\textbf{(iii) Compatibility with Dynamic Programming.}
Although $\mathcal P(\mathcal G^\mathcal X_\ast)$ is formally infinite-dimensional, for finite-horizon problems the relevant information is localized. As shown in Section~\ref{sec:rl}, the optimal policy at time $t$ depends only on $(T-t)$-hop neighborhoods, yielding a tractable DPP on progressively smaller lifted state spaces.

Viewed this way, graph-lifted MFC as we propose it in this work does not replace classical MFC but extends it; the mean field is enriched just enough to encode the interaction structure required while preserving the conceptual simplicity of a single representative control problem.

\subsection{Sparse Random Graphs and Local Weak Convergence}
Our work focuses on systems where agents are nodes in a large sparse graph, contrasting against simpler dense graphs where a MF may be computed for neighbors of a node. A canonical example is the \Erdos-\Renyi\ random graph $G(N, d/N)$. As $N \to \infty$, such graphs are \textit{locally tree-like}, formalized by the concept of \textbf{local weak convergence} or Benjamini-Schramm convergence \cite{benjamini2001recurrence}. By a $k$-neighborhood of a node $\rho$ in $\mathcal G$ we mean the set
$\left\{ v \in \mathcal G: \Vert v - \rho \Vert \le k \right\}$,
where the norm is defined through the graph distance. We will also call a $k$-neighborhood a $k$-hop neighborhood occasionally and denote it by $(\mathcal G, \rho)_k$. 
\begin{definition}
    Let $\mathcal{G}_*$ be the space of rooted graphs. A sequence of finite, random graphs $(G_N)_{N \ge 1}$ converges to a random rooted graph $(\mathcal{G}, \rho)$ if for every fixed radius $k \ge 1$, the distribution of the $k$-hop neighborhood of a uniformly chosen vertex in $G_N$ converges to the distribution of the $k$-hop neighborhood of the root $\rho$ in $\mathcal{G}$. 
\end{definition}
This implies that for a sequence of large sparse graphs $(G_N)_{N \in \mathbb N}$ that converges weakly locally to $(\mathcal G, \rho)$, agents typically experience a \textit{distribution} of local environments.

\section{The Sparse Mean-Field Control Framework}
\label{sec:sparse_framework}
We now formalize the control problem on a large sparse graph over a finite time horizon $T$. In all what follows we write $\mathcal{P}(\Omega)$ to denote the space of probability measures over the set $\Omega$.
We also require the notion of a decorated graph. Throughout, all spaces considered will be assumed to be (at least) standard Borel spaces.
\begin{definition}[Decorated Rooted Graph]
A decorated rooted graph is a pair $((\mathcal{G}, \rho), \mathbf{x})$, where $(\mathcal{G}, \rho)$ is a rooted graph and $\mathbf{x}: V(\mathcal{G}) \to \mathcal{X}$ is a state decoration. Let $\mathcal{G}_*^{\mathcal{X}}$ be the space of all such objects. For a radius $k$, let $\mathcal{G}_{*,k}^{\mathcal{X}}$ be the space of decorated $k$-neighborhoods.
\end{definition}

\subsection{The Finite N-Agent System}
Let $G_N = (V_N, E_N)$ be a graph with $|V_N|=N$ agents. Let $((G_N, i), \mathbf{x}_t)_k$ denote the $k$-hop decorated neighborhood of agent $i$ at time $t$. The system is modeled as follows.
\begin{itemize}[noitemsep,topsep=0pt]
    \item \textbf{State}: Each agent $i \in V_N$ has a state $\mathbf x_t^i \in \mathcal{X}$. The joint state is $\mathbf{x}_t = (\mathbf x_t^i)_{i \in V_N} \in \mathcal{X}^N$. The initial states are given as (random) boundary conditions $\mathbf x_0$.
    \item \textbf{Action}: Each agent $i$ takes an action $u_t^i \in \mathcal{U}$.
    \item \textbf{Policy}: Agents follow a shared, local policy $\pi$. The action $u_t^i$ depends on the agent's local decorated neighborhood. For simplicity of exposition, consider nearest-neighbor interactions ($k=1$):
    \begin{equation}
        \label{equ:action_kernel_sample}
        u_t^i \sim \pi_t( \mathrm du_t^i \mid (((G_N, i), \mathbf{x}_t)_1).
    \end{equation}
    More generally, we define a policy as follows.
\begin{definition}[Policy]
A policy $\pi_t$ is a stochastic kernel $\pi_t: \mathcal{G}_*^{\mathcal{X}} \to \mathcal{P}(\mathcal{U})$ that provides an action $u_t \sim \pi_t(\cdot \mid ((\mathcal{G}, \rho), \mathbf{x}))$ for node $\rho$.
\end{definition}
    \item \textbf{Dynamics}: The state of each agent evolves based on its local neighborhood:
    \begin{equation}
        \label{equ:state_kernel_sample}
        \mathbf x_{t+1}^i \sim P( \mathrm d \mathbf x_{t+1}^i \mid ((G_N, i), \mathbf{x}_t)_1, u_t^i).
    \end{equation}
    \item \textbf{Objective}: Maximize the total expected reward:
    \begin{equation}
        \label{equ:finite_agent_rewards}
        J^N(\pi) = \E \left[ \sum_{t=0}^{T-1} \frac{1}{N} \sum_{i=1}^N r(((G_N, i), \mathbf{x}_t)_1, u_t^i) \right]
    \end{equation}
    Moreover, we assume $r: \mathcal G^{\mathcal X}_{*,k} \times \mathcal U \to \mathbb R$ to be continuous and bounded.
\end{itemize}
The above setting describes a fully observable MDP, i.e. a standard RL problem, when we optimize over all policies. This is reflected in the fact that the joint states and actions of agents may be equivalently considered, as the MDP's states and actions. For tractability however, in standard MFC it is fruitful to introduce a ``system state`` given by a probability measure $\mu \in \mathcal P ( \mathcal X)$, as well as to look for feedback policies shared by all agents \cite{saldi2018markov} that are optimal under a certain class of policies. However, such a simple system state is insufficient in the setting of sparse graphs. The above shared-policy setting is the one we will study for the remainder of this paper, developing the suitable sparse mean-field limit for solving the RL problem.

\subsection{The Limiting Sparse-MFC System}
As $N \to \infty$, we assume $G_N$ converges to a random rooted graph $(\mathcal{G}, \rho)$ in the Benjamini--Schramm sense. 
The state of our Sparse-MFC (SMFC) system at time $t$ is therefore naturally given by a probability measure $\boldsymbol{\mu}_t \in \mathcal{P}(\mathcal{G}_*^{\mathcal{X}})$.

By formally taking $N \to \infty$ in \eqref{equ:finite_agent_rewards}, the normalized sum  converges to an integral over the space of decorated graphs, distributed according to some law determined by the Benjamini--Schramm limit. The evolution of the state is governed by a transition operator $\boldsymbol{\mu}_{t+1} = T_t(\boldsymbol{\mu}_t, \pi_t)$. Indeed, it is just the result of an application of the two kernels according to \eqref{equ:action_kernel_sample} and \eqref{equ:state_kernel_sample}. One averages these dynamics according to the current system distribution, to obtain
\begin{align}
    \boldsymbol{\mu}_{t+1}(\mathrm d S) = \int \boldsymbol \mu_t (\mathrm d y) \int \Pi_t(\mathrm d \overline u | y) \overline P(\mathrm d S | y, \overline u).
\end{align}
Here, $\overline u = (u^{v_1},...)$ for some enumeration of $V(\mathcal G)$ and
$$\Pi_t(\mathrm d \overline u | S) = \prod\limits_{v \in \mathcal G} \pi_t ( \mathrm d u^v | (S,v)_1).$$
Similarly, $\overline P$ uses the state and the local action to sample the next state for every agent $v \in V(\mathcal G)$.
Note that the evolution operators $(T_t)_{t \in \mathbb N}$ yield a deterministic flow of probability measures in case the initial state $\boldsymbol{\mu}_0$ and the policy $\pi$ is specified.
We will usually omit explicit dependence on $\pi$ when clear from context.

The objective functional in the limit is given with $S = (\mathcal{G}, \rho)$ being a rooted graph, by
\begin{multline}
    \label{equ:full_limiting_objective}
    J(\pi) = \E \Big[ \sum_{t=0}^{T-1} \int_{\mathcal{G}_*^{\mathcal{X}}} \int_{\mathcal{U}} r((S, \mathbf{x}_t)_1, u_t) \\ \times \pi_t(\mathrm du_t \mid S) \, \boldsymbol{\mu}_t(\mathrm dS) \Big]
\end{multline}
where it is important that the reward only takes local neighborhoods, meaning that taking $N \to \infty$ indeed yields a sensible limit by assumption. 

The value function of a sequence of policies $\pi = (\pi_t)_{t \le T}$ with boundary conditions $\mathbf x_0$ is defined, with $0 < \gamma \le 1$, by
\begin{multline}
    \label{equ:value_fct_def}
    V_t^{\pi}(\boldsymbol \mu) = \mathbb E ^\pi \Big[ \sum\limits_{t'=t} ^{T-1} \gamma ^{t'} \int \boldsymbol{\mu}_{t'}(\mathrm d (\mathcal G,\rho) )  \\
    \times r \big( ((\mathcal G,\rho),\mathrm x_{t'})_1,u_{t'} \big) \Big],
\end{multline}
where the expectation w.r.t. $\pi$ is defined such that
$$\boldsymbol \mu_{t+1} = T_{t+1} (\boldsymbol \mu_t),$$
$$u^i_t \sim \pi_t \big(\mathrm d u^i_t \mid \big((\mathcal G,i), \mathbf x_t \big)_1 \big),$$
$$ \mathbf x_{t+1} ^i \sim P(\mathrm d \mathbf x_{t+1} ^i \mid \big((\mathcal G,i), \mathbf x_t \big)_1,u_t ^i ).$$
Finally, we abbreviate $V(\boldsymbol \mu) = \sup_\pi V^\pi (\boldsymbol \mu)$.
\subsection{Approximation Guarantee}
The practical value of this framework depends on its ability to approximate the finite $N$-agent system, which we obtain with rigorous guarantees in the following.
\begin{theorem}[Convergence of Optimal Values] \label{thm:convergence}
Let $(G_N)_{N \ge 1}$ be a sequence of graphs converging locally to $(\mathcal{G}, \rho)$.
Assume that the space of admissible policies $\Pi$ is compact and that the family $(J^N)_{N\in\mathbb N}$
is equicontinuous (with respect to the topology of weak convergence of policies) on $\Pi$.
Then
\begin{equation}
    \lim_{N \to \infty} \sup_{\pi} J^N(\pi) \;=\; \sup_{\pi} J(\pi).
\end{equation}
\end{theorem}

\begin{proof}
Fix an admissible policy $\pi$. By Theorem~\ref{thm:locality} and \citet[Theorem~3.2]{lacker2023local},
we have $J^N(\pi)\to J(\pi)$ as $N\to\infty$.
Next, let $V_N^*=\sup_{\pi}J^N(\pi)$ and $V^*=\sup_{\pi}J(\pi)$.
By assumption, the maps $\pi\mapsto J^N(\pi)$ and $\pi\mapsto J(\pi)$
are continuous; hence maximizers exist on the compact set of admissible policies.
Pick $\pi^*\in\arg\max_\pi J$ and $\pi_N\in\arg\max_\pi J^N$.
Then $V_N^*\ge J^N(\pi^*)\to J(\pi^*)=V^*$, so $\liminf_N V_N^*\ge V^*$.
By compactness, along a subsequence $\pi_N\Rightarrow\bar\pi$ and by continuity,
$V_N^*=J^N(\pi_N)\to J(\bar\pi)\le V^*$, hence $\limsup_N V_N^*\le V^*$.
Therefore $V_N^*\to V^*$.
\end{proof}

\begin{corollary}[Approximate Optimality]
An optimal policy $\pi^*$ for the SMFC problem is $\epsilon$-optimal for the $N$-agent problem, with $\epsilon \to 0$ for sufficiently large $N$.
\end{corollary}

\section{Reinforcement Learning on Sparse Graphs}
\label{sec:rl}
In what follows, we will develop an SMFC-based RL algorithm applicable to sparse graphs, that will naturally fit with GNNs, and is theoretically well-founded through solving a limiting MDP, whose foundation is proved in the following.

\subsection{Horizon-Dependent Local Dynamic Programming}
To begin, a key insight is that for a finite horizon, we do not need to consider the distribution over infinite graphs.

\begin{theorem}[Horizon-Dependent Locality] \label{thm:locality}
Consider the SMFC problem with horizon $T$. The value function $V_t$ depends only on the distribution of $(T-t)$-hop decorated neighborhoods, i.e., $V_t: \mathcal{P}(\mathcal{G}_{*, T-t}^{\mathcal{X}}) \to \mathbb{R}$. In other words, $V_t$ is measurable w.r.t. the information contained in $\mathcal{G}_{*, T-t}^{\mathcal{X}}$. Consequently,
the optimal policy $\pi_t^*$ at time $t$ depends only on the $(T-t)$-hop decorated neighborhood of the root.
\end{theorem}

\begin{remark}
Theorem \ref{thm:locality} is crucial for computation. It implies that for a $T$-horizon problem, the state space is not $\mathcal{P}(\mathcal{G}_*^{\mathcal{X}})$ but $\mathcal{P}(\mathcal{G}_{*,T}^{\mathcal{X}})$. On a regular graph like a grid or a $d$-regular tree, the number of non-isomorphic $k$-hop neighborhoods is finite. This means the state space $\mathcal{P}(\mathcal{G}_{*,k}^{\mathcal{X}})$ becomes computationally tractable. 
\end{remark}

This locality result leads to a well-defined DPP/MDP on a sequence of progressively smaller state spaces.

\begin{theorem}[Dynamic Programming Principle] \label{thm:dpp}
The value function $V_t: \mathcal{P}(\mathcal{G}_{*, T-t}^{\mathcal{X}}) \to \mathbb{R}$ satisfies the Bellman equation:
\begin{equation}
\label{equ:dpp}
\begin{cases}
    V_T(\cdot) = 0, \\[0.3em]
    \begin{alignedat}{2}
        V_t(\boldsymbol{\mu})
        &\,=\, \sup_{\pi_t} \Bigl\{
        &&\mathbb{E}_{S \sim \boldsymbol{\mu}}
          \bigl[
            \mathbb{E}_{u \sim \pi_t(\cdot \mid S)}[r]
          \bigr] \\
        &&&+ V_{t+1}\bigl(T_t(\boldsymbol{\mu}, \pi_t)\bigr)
        \Bigr\},
        \qquad t < T .
    \end{alignedat}
\end{cases}
\end{equation}

where $\boldsymbol{\mu} \in \mathcal{P}(\mathcal{G}_{*, T-t}^{\mathcal{X}})$ and $T_t$ maps to a distribution over $(T-t-1)$-hop neighborhoods. An optimal policy $\pi_t^*$ exists.
\end{theorem}

This DPP gives rise to a natural solution as an MDP and via using RL, such as value-based or policy-based RL. In particular, note that the MDP is however at least continuous or in the worst case infinite-dimensional, since the action space is now given by the initial policy space of the finite graph problem. Therefore, in the following we will formulate a policy gradient algorithm based on this limiting SMFC problem, since the value-based approach is intractable.

\subsection{Policy Gradients on Sparse Graphs}
To provide a solid foundation for the proposed RL algorithm, we derive a policy gradient theorem for the sparse mean-field objective. We treat the problem as a "Lifted MDP" where the state is the neighborhood distribution $\boldsymbol{\mu}$.

To handle the complexity of the space, we employ a hierarchical MFC policy structure:
\begin{enumerate}[noitemsep,topsep=0pt]
    \item A \textbf{Meta-Policy} $\hat{\pi}_\theta(\psi \mid \boldsymbol{\mu})$ parameterized by $\theta$, which observes the global mean-field state $\boldsymbol{\mu}$ and samples a local policy parameter $\psi \in \Psi$.
    \item A \textbf{Local Policy} $\pi_\psi(u \mid s)$ parameterized by $\psi$, which agents use to select actions $u$ based on their local neighborhood $s$.
\end{enumerate}

\begin{assumption}[Regularity]
\label{ass:policy-domination}
The meta-policy $\hat{\pi}_\theta(\psi \mid \boldsymbol{\mu})$ is differentiable with respect to $\theta$ and satisfies the standard score function requirements. Specifically, $\nabla_\theta \log \hat{\pi}_\theta(\psi \mid \boldsymbol{\mu})$ exists and is bounded. The reward $r$ and transition dynamics are bounded.
\end{assumption}

Fix the horizon $T$. Let $\boldsymbol{\mu}_0$ be the initial distribution. The evolution of the system is deterministic at the level of distributions given the chosen parameters $\psi_t$:
\[ \boldsymbol{\mu}_{t+1} = T_t(\boldsymbol{\mu}_t, \pi_{\psi_t}). \]
However, since the meta-policy $\hat{\pi}_\theta$ is stochastic, the trajectory of parameters $\tau = (\psi_0, \dots, \psi_{T-1})$ is random, inducing a probability distribution over the trajectories of $\boldsymbol{\mu}_t$.

\begin{theorem}[Policy Gradient for the Limiting SMFC]
\label{thm:policy-gradient}
Consider the limiting SMFC problem formulated as a Lifted MDP on the state space $\mathcal{P}(\mathcal{G}_*^{\mathcal{X}})$. The gradient of the objective $J(\theta) = \mathbb{E}[\sum_{t=0}^{T-1} R(\boldsymbol{\mu}_t, \psi_t)]$ with respect to the meta-parameters $\theta$ is given by:
\begin{equation}
    \nabla_\theta J(\theta) = \sum_{t=0}^{T-1} \mathbb{E}_{\boldsymbol{\mu}_t, \psi_t} \left[ \nabla_\theta \log \hat{\pi}_\theta(\psi_t \mid \boldsymbol{\mu}_t) \cdot Q^{\text{MF}}_t(\boldsymbol{\mu}_t, \psi_t) \right],
\end{equation}
where $Q^{\text{MF}}_t(\boldsymbol{\mu}, \psi)$ is the cumulative future reward starting from state $\boldsymbol{\mu}$ with action $\psi$:
\begin{equation}
    Q^{\text{MF}}_t(\boldsymbol{\mu}, \psi) = R(\boldsymbol{\mu}, \psi) + \mathbb{E}_{\tau > t} \left[ \sum_{k=t+1}^{T-1} R(\boldsymbol{\mu}_k, \psi_k) \right],
\end{equation}
with $R(\boldsymbol{\mu}, \psi) \coloneqq \int r((s)_1, u) \pi_\psi(\mathrm du|s) \boldsymbol{\mu}(\mathrm ds)$.
\end{theorem}

The proof is provided in Appendix \ref{app:proof_pg}. We note that the form of this gradient is a standard result for MDPs; the novelty lies in the definition of the state space $\boldsymbol{\mu}$ and the rigorous limit justifying it.

\noindent\textbf{Finite-Graph Policy Gradient Approximation.}
The core practical challenge is that we cannot compute the gradient on the infinite limit directly. Instead, we compute the gradient on a finite graph $G_N$. Explicitly, the gradient we compute in Algorithm \ref{alg:smfac} is:
\begin{equation}
    \label{eq:finite_grad}
    \nabla_\theta J^N(\theta) = \mathbb{E} \left[ \sum_{t=0}^{T-1} \left( \sum_{k=t}^{T-1} \bar{r}^N_k \right) \nabla_\theta \log \hat{\pi}_\theta(\psi_t \mid \boldsymbol{\mu}^N_t) \right],
\end{equation}
where $\boldsymbol{\mu}^N_t$ is the empirical distribution of decorated neighborhoods in $G_N$, and $\bar{r}^N_k = \frac{1}{N} \sum_{i=1}^N r(s^i_k, u^i_k)$ is the average reward at time $k$.

We must establish that this finite computation approximates the true gradient of the underlying SMFC MDP.

\begin{theorem}[Finite-Graph PG Approximation] \label{thm:ctde}
    Let $J^N(\theta)$ be the objective on the finite graph $G_N$ and $J(\theta)$ be the limiting objective. Assume the transition operator $T_t$ and reward $r$ are Lipschitz continuous with respect to the weak topology on measures. Then, centralized training on the finite graph asymptotically optimizes the infinite limit:
    \begin{align*}
        \lim_{N \to \infty} \left\Vert \nabla_\theta J^N(\theta) - \nabla_\theta J(\theta) \right\Vert = 0.
    \end{align*}
\end{theorem}
This validates Alg.~\ref{alg:smfac}: we can use standard PPO on the global graph state to optimize the mean-field control problem.

We have thus motivated a practical reinforcement learning approach
based on local function approximation.
The local policy $\pi_t$ takes a decorated neighborhood as input. Using a GNN for $\hat \pi$, the resulting algorithm is Algorithm~\ref{alg:smfac}, with its application to graphs visualized in Figure~\ref{fig:architecture}.

\begin{figure*}[t]
  \centering
  \resizebox{0.9\linewidth}{!}{
    \begin{tikzpicture}[
        node distance=1.0cm,
        font=\sffamily,
        agent/.style={circle, draw=nodeblue, fill=white, thick, minimum size=0.4cm, inner sep=0pt},
        root/.style={circle, draw=rootred, fill=rootred!20, thick, minimum size=0.5cm, inner sep=0pt},
        embedding/.style={rectangle, draw=gray!80, fill=green!10, minimum width=0.2cm, minimum height=0.6cm, rounded corners=1pt},
        globalvec/.style={rectangle, draw=rootred, fill=rootred!10, minimum width=0.4cm, minimum height=1.2cm, rounded corners=2pt, thick},
        process/.style={rectangle, draw=black, fill=white, rounded corners, minimum height=1cm, minimum width=2cm, align=center, drop shadow},
        arrow/.style={->, >=Stealth, thick, color=gray!40, line width=1.5mm},
        thinarrow/.style={->, >=Stealth, thick, color=black},
        labeltext/.style={font=\small\sffamily, align=center, color=black!80}
    ]

    \node[labeltext, font=\bfseries] (step1) {1. Input State};
    
    \begin{scope}[shift={(0,1.5)}, below=0.2cm of step1, local bounding box=inputGraph]
        \node[root] (r) at (0,0) {};
        \node[agent] (n1) at (0.8, 0.5) {};
        \node[agent] (n2) at (-0.6, 0.7) {};
        \node[agent] (n3) at (-0.5, -0.6) {};
        \draw[thick, linkgray] (r) -- (n1);
        \draw[thick, linkgray] (r) -- (n2);
        \draw[thick, linkgray] (r) -- (n3);
        \draw[thick, linkgray] (n1) -- (n2);
        
        \node[below right=0.1cm and -0.3cm of r, font=\scriptsize] {$G_N, \mathbf{x}_t$};
    \end{scope}

    \node[process, right=1.2cm of inputGraph] (gnn) {\textbf{GNN Encoder} ($\Phi_\theta$)\\ \scriptsize Message Passing ($k=T-t$)};
    
    \node[embedding, right=0.5cm of gnn, label=below:\scriptsize $h_1$] (e1) {};
    \node[embedding, right=0.1cm of e1, label=below:\scriptsize $h_2$] (e2) {};
    \node[embedding, right=0.1cm of e2, label=below:\scriptsize $\dots$] (e3) {};
    \node[embedding, right=0.1cm of e3, label=below:\scriptsize $h_N$] (e4) {};
    
    \draw[arrow] (inputGraph) -- (gnn);
    \draw[thinarrow] (gnn) -- (e1);

    \node[trapezium, draw=gray, fill=highlightbg, rotate=-90, minimum width=1.5cm, minimum height=1.0cm, right=1.5cm of e4, anchor=north] (pool) {};
    \node[align=center, font=\scriptsize, rotate=0] at (pool.center) {Mean\\Pool};
    
    \node[globalvec, above right=-0.1cm and 1.2cm of pool] (global) {};
    \node[below=0.1cm of global, labeltext] {Meta-Policy \\Parameters};

    \draw[thinarrow] (e4) -- (pool);
    \draw[thinarrow] (pool) -- (global);

    \node[process, right=1.0cm of global, fill=nodeblue!10] (mlp) {\textbf{Meta-Policy} $\hat{\pi}$\\ \scriptsize (MLP)};
    
    \draw[arrow] (global) -- (mlp);

    \node[circle, draw=black, fill=white, right=1.5cm of mlp, minimum size=1.8cm, align=center] (dist) {$\pi_t$};
    \node[above=0.1cm of dist, labeltext] {Sampled\\Local Policy};
    
    \draw[arrow] (mlp) -- (dist);

    
    \begin{scope}[shift={(5,-1.2)}, below=2.0cm of mlp]
        \node[root, label=left:\scriptsize Agent $i$] (agent_i) {};
        \node[agent] (nei1) at (0.6, 0.3) {};
        \node[agent] (nei2) at (0.4, -0.5) {};
        \draw[thick, linkgray] (agent_i) -- (nei1);
        \draw[thick, linkgray] (agent_i) -- (nei2);
        
        \node[right=1.5cm of agent_i, rectangle, draw=black, dashed, rounded corners, minimum height=1cm, align=center] (apply) {Execute\\$u_t^i \sim \pi_t(\cdot | \text{neigh}_i)$};
        
        \draw[thinarrow] (agent_i) -- (apply);
        
        \draw[thinarrow, dashed, red] (dist.south) |- (apply.east) node[midway, right, font=\scriptsize, color=red] {Parameters};
    \end{scope}
    
    \node[above=0.2cm of gnn, font=\bfseries, color=nodeblue] {1. Encoding};
    \node[above=1.2cm of pool, font=\bfseries, color=nodeblue] {2. Aggregation};
    \node[above=0.2cm of mlp, font=\bfseries, color=nodeblue] {3. Policy Gen.};
    
    \node[below=1.1cm of pool, font=\tiny, color=gray, align=center] {Approximates Integral\\Eq. \ref{eq:gnn_integration}};

    \end{tikzpicture}
  }
  \caption{\textbf{The Sparse Mean-Field RL Architecture.} The diagram illustrates the computational flow of Algorithm \ref{alg:smfac}. (1) The GNN encodes the local structure of the entire graph $G_N$. (2) A global pooling operation aggregates these embeddings to obtain a function of the mean-field state $\boldsymbol{\mu}_t$. (3) The Meta-Policy (Actor) uses this mean-field state-dependent output to generate the parameters for the local policy $\pi_t$. (Bottom) Individual agents $i$ then execute actions $u_t^i$ by conditioning this shared policy $\pi_t$ on their specific local neighborhoods.}
  \label{fig:architecture}
\end{figure*}
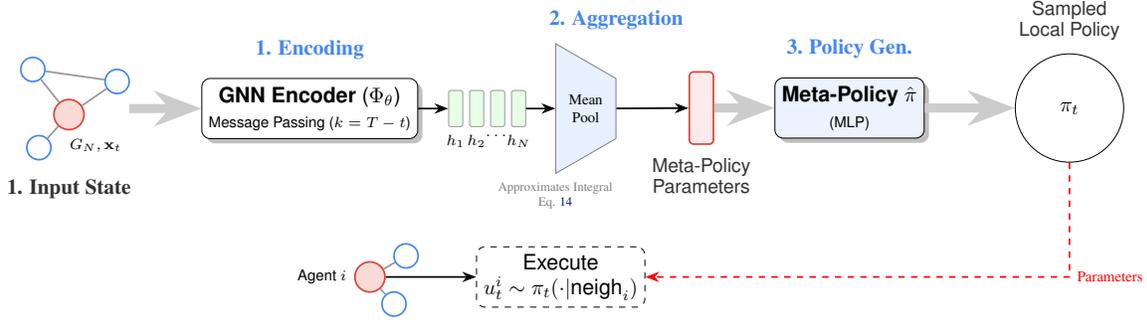

\begin{algorithm}[b!]
    \caption{Sparse Mean-Field Actor--Critic (SMFAC)}
    \label{alg:smfac}
    \begin{algorithmic}[1]
        \STATE Initialize actor GNN $\pi_\theta$ and critic $V_\phi$.
        \FOR {each training episode}
            \STATE Initialize $N$-agent system $(G_N, \mathbf{x}_0)$.
            \FOR {time $t = 0, \ldots, T-1$}
                \STATE Apply GNN to $G_N$ with $k=T - t$ message passes.
                \STATE Sample $\pi _t \sim \mathrm{MLP} \left( \frac 1 N \sum_{i=1} ^N \mathrm{GNN}_k ^i (G_N) \right).$
                \FOR{agent $i \in G_N$}
                    \STATE Sample $u_t^i \sim \pi_t(\cdot \mid \big((G_N,i), \mathbf x_t \big)_1)$.
                \ENDFOR
                \STATE Execute all actions, observe $\mathbf{x}_{t+1}$ and rewards $\mathbf{r}_t$.
                \STATE Store transitions in a replay buffer.
            \ENDFOR
            \STATE Compute advantages; update $\pi_\theta$ and $V_\phi$ using PPO.
        \ENDFOR
    \end{algorithmic}
\end{algorithm}

\subsection{Theoretical Analysis of GNN Usage}
\label{sec:theory_errors}
Our framework relies on two key approximations: (1) using a finite radius $k$ (via GNN depth) to approximate the optimal policy which formally requires radius $T-t$, and (2) using a finite graph $G_N$ to approximate the infinite limit. Here we provide theoretical bounds for both.

\subsubsection{Approximation by Truncated Policies.}
While Theorem \ref{thm:locality} guarantees that radius $T-t$ is sufficient, practical GNNs often use a fixed, smaller depth $k \ll T$. We show that if the dynamics are in some suitable sense not too much influenced by a single decoration in the $1$-hop neighborhood, then the approximation error is exponentially small. For the exact assumptions, we refer the reader to the Appendix.

\begin{theorem}[Truncation Error]
\label{thm:truncation}
Let $\pi^*$ be the optimal policy (requiring radius $T$) and $\pi^{(k)}$ be the optimal policy constrained to use only $k$-hop information. Then the loss in performance is bounded by
\begin{equation}
    |J(\pi^*) - J(\pi^{(k)})| \le O(\lambda^k),
\end{equation}
where $\lambda$ is explicitly given in terms of the maximum degree of the underlying graph and "smoothness" of dynamics $P$.
\end{theorem}
This result justifies the use of shallow GNNs even for problems with longer horizons, provided graph interactions are not pathological (e.g., absence of critical phase transitions).

\subsubsection{GNNs as Universal Mean-Field Approximators}

A central theoretical question is whether Graph Neural Networks are capable of representing the functions required for Sparse MFC, specifically the value function $V(\boldsymbol{\mu})$ and the meta-policy $\hat{\pi}(\boldsymbol{\mu})$. Since the state $\boldsymbol{\mu}$ is a probability measure over the space of decorated rooted graphs $\mathcal{G}_{*,k}^{\mathcal{X}}$, any approximator must satisfy two key geometric properties: (1) permutation invariant with respect to agent indices (measure property), and (2) isomorphism invariant with respect to the local neighborhood structure (the graph property).

We analyze this through the lens of Geometric Deep Learning \cite{bronstein2021geometric}, decomposing the approximation into a global aggregation step and a local encoding step.

\noindent\textbf{Global Aggregation (Deep Sets).}
The value function $V: \mathcal{P}(\mathcal{G}_{*,k}^{\mathcal{X}}) \to \mathbb{R}$ is a function on a set (or measure). The Deep Sets theorem \cite{zaheer2017deep} establishes that any continuous permutation-invariant function on a set can be decomposed into the form $\rho(\sum \phi(x))$. In our mean-field context, this implies that the value function can be approximated as:
\begin{equation}
    V(\boldsymbol{\mu}) \approx \rho \left( \int_{\mathcal{G}_{*,k}^{\mathcal{X}}} \phi(\mathfrak{g}) \, \mathrm{d}\boldsymbol{\mu}(\mathfrak{g}) \right),
\end{equation}
where $\phi$ is a continuous embedding of a rooted graph, and $\rho$ is a continuous function (e.g., an MLP).

\noindent\textbf{Local Encoding (Weisfeiler-Leman).}
The function $\phi(\mathfrak{g})$ must encode the topology of the rooted neighborhood $\mathfrak{g}$. It is established that Message Passing Neural Networks (MPNNs) are as powerful as the 1-dimensional Weisfeiler-Leman (1-WL) graph isomorphism test \cite{xu2018how,morris2019weisfeiler,maron2019provably}. Therefore, parameterizing $\phi$ as a GNN allows us to distinguish all neighborhood structures distinguishable by the 1-WL test.

\begin{proposition}[GNN Readout as Mean-Field Integration]
Let $\Phi_\theta: \mathcal{G}_{*,k}^{\mathcal{X}} \to \mathbb{R}^d$ denote a $k$-layer MPNN parameterized by $\theta$. The operation of applying the GNN to the global graph $G_N$ and performing global mean pooling is mathematically equivalent to the Deep Sets decomposition over the empirical measure $\boldsymbol{\mu}^N_t$:
\begin{equation}
    \label{eq:gnn_integration}
    \underbrace{\frac{1}{N} \sum_{i=1}^N \Phi_\theta\bigl( (G_N, i)_k \bigr)}_{\text{Global Average Pooling}} 
    \;=\; 
    \underbrace{\int_{\mathcal{G}_{*,k}^{\mathcal{X}}} \Phi_\theta(\mathfrak{g}) \, \boldsymbol{\mu}^N_t(\mathrm{d}\mathfrak{g})}_{\text{Expectation over Mean Field}}.
\end{equation}
\end{proposition}

This identity provides the rigorous justification for our algorithmic choice. It confirms that the standard "GNN + Readout" architecture is not merely a heuristic, but a direct implementation of the function approximation required for the Mean-Field state space.
Consequently, we approximate the global value function (Critic) and similarly the meta-policy as:
\begin{equation}
    V_\phi(\boldsymbol{\mu}^N_t) \approx \rho_\phi \left( \frac{1}{N} \sum_{i=1}^N \Phi_{\theta_{\text{enc}}}\bigl( (G_N, i)_k \bigr) \right).
\end{equation}

\begin{remark}[Expressivity Limitations]
While this framework is robust, it inherits the limitations of the 1-WL test. If the optimal control depends on distinguishing local structures that are 1-WL equivalent (e.g., certain circular skip-connections vs. disconnected components), standard MPNNs may underfit. In such cases, the framework naturally extends to higher-order GNNs (k-GNNs) \cite{morris2019weisfeiler} without altering the mean-field logic.
\end{remark}

\section{Experiments}
We validate the necessity of the proposed sparse mean-field framework on two epidemic control scenarios. We compare against the strong but heuristic LWMFMARL baseline \cite{fabian2025learning} which conditions actions on global population statistics and node degrees. While already very effective in many settings with millions of nodes, this baseline is heuristic and relies on aggregated representations, discarding fine-grained structural information. We rigorously generalize this heuristic baseline. The experiments isolate two complementary failure modes of mean-field control on sparse graphs and demonstrate how local, graph-aware policies resolve them. The setup is found in Appendix~\ref{app:exp}.

\subsection{Resolving Local Heterogeneity}
In supercritical infection regimes where infection costs dominate, purely mean-field/LWMFMARL controllers converge to a degenerate \emph{vaccinate-all} policy. More generally, because they only observe aggregated statistics, they cannot distinguish locally exposed nodes from spatially isolated ones. On sparse graphs, this behavior is highly suboptimal, as vaccinating a small number of strategically placed agents is often sufficient to block transmission.
\begin{figure}[t]
    \centering
    \includegraphics[width=1\linewidth]{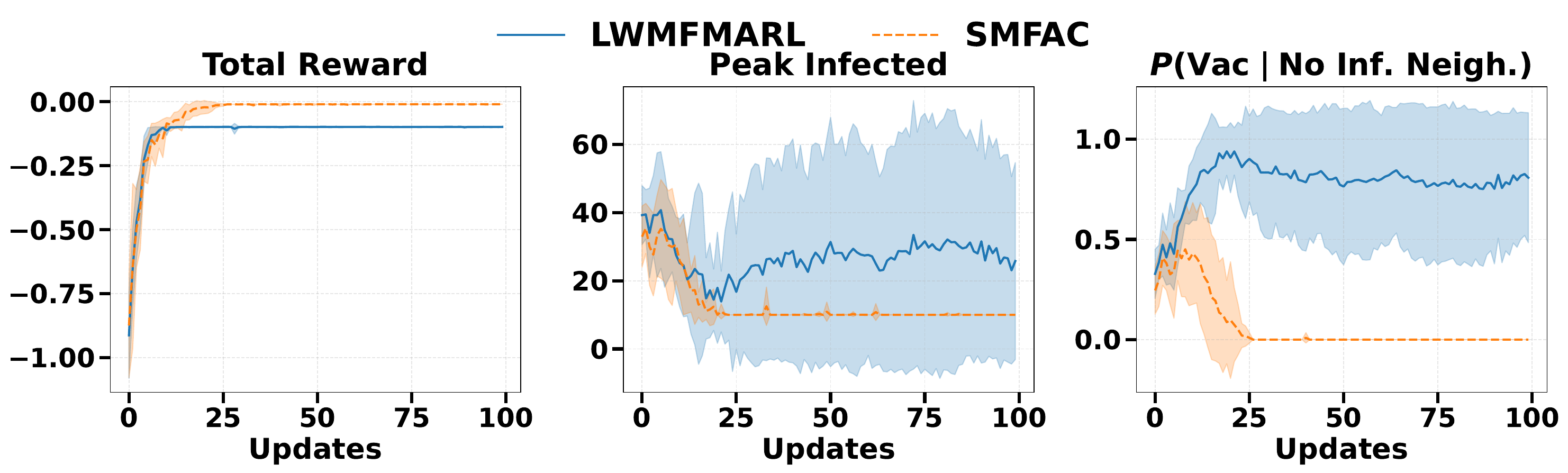}
    \caption{The additional information at the local node makes it easier for the agents to decide when to vaccinate, yielding a better policy compared to only local information.}
    \label{fig:grid_loss}
\end{figure}

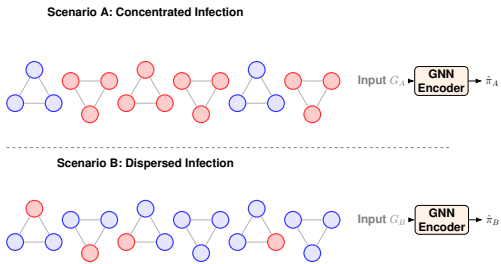
\begin{figure}[b]
\centering
\resizebox{0.8\columnwidth}{!}{%
\begin{tikzpicture}[
  node distance=1.5cm,
  susceptible/.style={circle, draw=blue!80, fill=blue!10, very thick, minimum size=0.6cm, inner sep=0pt},
  infected/.style={circle, draw=red!80, fill=red!20, very thick, minimum size=0.6cm, inner sep=0pt},
  edge_style/.style={-, thick, draw=gray!60},
  label_style/.style={font=\sffamily\bfseries\large, align=center}
]

\tikzset{every node/.append style={font=\sffamily\bfseries\large}}


\newcommand{\drawTriangle}[4]{%
  \begin{scope}[shift={(#1,#2)}]
    \foreach \state [count=\i] in {#4} {%
      \ifnum\i=1 \def\ang{90}\fi
      \ifnum\i=2 \def\ang{210}\fi
      \ifnum\i=3 \def\ang{330}\fi
      \node[\state] (t#3-\i) at (\ang:0.8) {};
    }%
    \draw[edge_style] (t#3-1) -- (t#3-2) -- (t#3-3) -- (t#3-1);
  \end{scope}%
}

\newcommand{\drawTriangleDown}[4]{%
  \begin{scope}[shift={(#1,#2)}]
    \foreach \state [count=\i] in {#4} {%
      \ifnum\i=1 \def\ang{270}\fi
      \ifnum\i=2 \def\ang{30}\fi
      \ifnum\i=3 \def\ang{150}\fi
      \node[\state] (t#3-\i) at (\ang:0.8) {};
    }%
    \draw[edge_style] (t#3-1) -- (t#3-2) -- (t#3-3) -- (t#3-1);
  \end{scope}%
}

\node[label_style] at (4, 2.9) {Scenario A: Concentrated Infection};

\drawTriangle{0}{0}{1}{susceptible, susceptible, susceptible}
\drawTriangle{4}{0}{2}{infected, infected, infected}
\drawTriangle{8}{0}{3}{susceptible, susceptible, susceptible}

\drawTriangleDown{2}{0}{7}{infected, infected, infected}
\drawTriangleDown{6}{0}{8}{infected, infected, infected}
\drawTriangleDown{10}{0}{9}{infected, infected, infected}

\node[right=0.5cm of t9-2, text=gray] (gnn_in) {Input $G_A$};
\node[rectangle, draw=black, fill=orange!10, right=0.2cm of gnn_in,
      align=center, rounded corners] (gnn) {GNN\\Encoder};
\node[right=0.5cm of gnn] (policy) {$\hat{\pi}_A$};
\draw[-Latex, thick] (gnn_in) -- (gnn);
\draw[-Latex, thick] (gnn) -- (policy);

\draw[dashed, gray] (-1, -2.0) -- (15, -2.0);

\node[label_style] at (4, -2.6) {Scenario B: Dispersed Infection};
\def\yoffset{-5}


\drawTriangle{0}{\yoffset}{4}{infected, susceptible, susceptible}
\drawTriangle{4}{\yoffset}{5}{susceptible, infected, susceptible}
\drawTriangle{8}{\yoffset}{6}{susceptible, susceptible, infected}

\drawTriangleDown{2}{\yoffset}{10}{infected, susceptible, susceptible}
\drawTriangleDown{6}{\yoffset}{11}{susceptible, susceptible, susceptible}
\drawTriangleDown{10}{\yoffset}{12}{susceptible, susceptible, susceptible}

\node[right=0.5cm of t12-2, text=gray] (gnn_in_b) {Input $G_B$};
\node[rectangle, draw=black, fill=orange!10, right=0.2cm of gnn_in_b,
      align=center, rounded corners] (gnn_b) {GNN\\Encoder};
\node[right=0.5cm of gnn_b] (policy_b) {$\hat{\pi}_B$};
\draw[-Latex, thick] (gnn_in_b) -- (gnn_b);
\draw[-Latex, thick] (gnn_b) -- (policy_b);

\end{tikzpicture}%
} 

\caption{Two infection scenarios with identical global counts but different local structure.}
\label{fig:two_scenarios}
\end{figure}

\begin{figure}[t]
    \centering
    \includegraphics[width=0.7\linewidth]{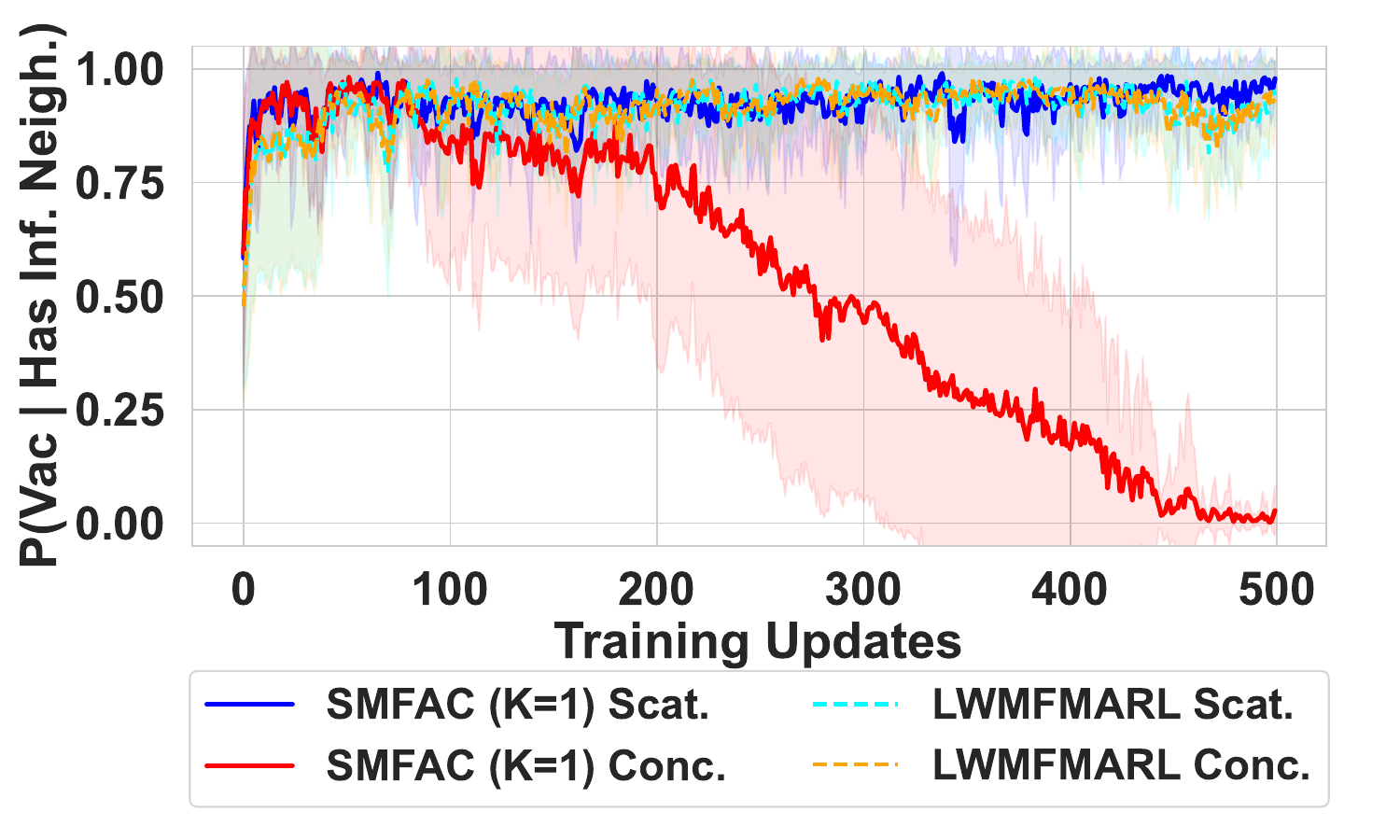}
    \caption{Applying a GNN encoder on the underlying graph can be essential in finding the optimal policy. Observe that our algorithm is able to find the optimal policy (not vaccinating nodes with infected neighbor if disease is already contained), where as the LWMFMARL is not able to distinguish between the concentrated and scattered initial state.}
    \label{fig:triangle_probs}
\end{figure}

We illustrate this effect on a two-dimensional grid. Figure~\ref{fig:grid_loss} shows that incorporating local neighborhood information via a lifted state space yields substantially lower loss and more selective interventions. Agents learn to condition actions on nearby infection states rather than reacting uniformly to global infection levels.

\subsection{Disambiguating Global Structure}
We next expose a more fundamental limitation of MFC and existing LWMFMARL: distinct local configurations can induce identical global statistics. Any controller that conditions solely on aggregated information is thus unable to distinguish such states, leading to suboptimal actions.

For simplicity, consider $N=20$ disjoint triangles with $5$ initially infected nodes. With equal probability, infections are either concentrated within a few triangles (Scenario~A in Figure~\ref{fig:two_scenarios}) or dispersed across many (Scenario~B), while global infection counts remain identical. A global penalty is incurred if more than a fixed number of triangles contain multiple infected nodes. When infections are dispersed, local containment is required; when they are concentrated, non-intervention is optimal. Since both cases induce the same MF statistics, the MF baseline must select the same action and is therefore suboptimal in at least one regime.

Figure~\ref{fig:triangle_probs} shows that the GNN-based policy successfully conditions its behavior on the underlying graph structure and separates the two regimes, whereas the baseline (equivalently, a GNN with zero message-passing steps) collapses to a single compromised policy.

\noindent\textbf{Experimental Takeaways.}
Across both experiments, MF aggregation fails either by obscuring localized infection pressure or by collapsing distinct configurations into identical global statistics. Neighborhood-aware policies resolve both issues by conditioning actions on local graph structure, improving MFC for sparse graphs.

\section{Conclusion and Discussion}

We have presented a rigorous framework for MFC on sparse graphs, overcoming the limitations of the classical dense interaction assumption by lifting the system state to distributions over decorated neighborhoods. Central to this approach is our horizon-dependent locality theorem, which establishes a tractable DPP and theoretically justifies the use of GNNs as local policy approximators. While this framework provides a solid theoretical grounding, practical challenges remain: the state space of neighborhood distributions can be high-dimensional, potentially increasing variance in gradient estimates. Furthermore, the algorithmic performance is inherently bounded by the expressivity of the GNN architecture (e.g., the Weisfeiler-Leman hierarchy).

Future research could extend this foundation to non-cooperative Sparse MFGs, coupling the Hamilton-Jacobi-Bellman equation for local agents with a Fokker-Planck-Kolmogorov equation for the neighborhood evolution to characterize Nash equilibria. Additionally, extending our finite-horizon results to the infinite-horizon setting offers a promising avenue; proving that the Bellman operator acts as a contraction on the space of neighborhood distributions would guarantee the existence of stationary optimal policies, further bridging the gap between mean-field theory and practical network control.

\clearpage
\noindent\textbf{Impact statement:} This paper presents work whose goal is to advance the field of Machine Learning. There are many potential societal consequences of our work, none which we feel must be specifically highlighted here.

\bibliographystyle{icml2026}
\bibliography{main}

\newpage
\appendix
\onecolumn
\section{Proof of Theorem \ref{thm:locality}}

\label{app:proof_locality}
We start with the following observation.
\begin{proposition}
\label{prop:kernel-randomization}
Let $(E,\mathcal E)$ be a measurable space and let $(S,\mathcal S)$ be a standard
Borel space. Let
\[
\pi : E \times \mathcal S \to [0,1]
\]
be a Markov kernel, i.e.\ for each $x\in E$, $\pi(x,\cdot)$ is a probability
measure on $(S,\mathcal S)$, and for each $A\in\mathcal S$, the map
$x \mapsto \pi(x,A)$ is $\mathcal E$-measurable.

Then there exist a measurable function
\[
f : E \times [0,1] \to S
\]
and a random variable $U \sim \mathrm{Unif}[0,1]$ such that for every fixed
$x\in E$,
\[
f(x,U) \sim \pi(x,\cdot).
\]
Equivalently, for all $A \in \mathcal S$,
\[
\mathbb P\bigl(f(x,U) \in A\bigr) = \pi(x,A).
\]
\end{proposition}
The proof of this Proposition is standard and can be found in slightly different phrasing in e.g. \cite{kallenberg1997foundations}. In order to formalize what we mean by ``the value function only depending on the local neighborhood``, we use the basic concept of measurability. By what we have just shown, every Markov kernel $\pi$ can be identified with a function $f_\pi$ and a random variable $U_\pi$, effectively factoring out the randomness into $U_\pi$. We can therefore associate a random variable $U_t ^v$ and $W_t ^v$ to the evolution of agent $v$ at time $t$, where $U_t ^v$ corresponds to the action choice of agent sitting at node $v$ at time $t$ (i.e., using $\pi_t$). Similarly, $W_t ^v$ describes the state evolution of agent $v$ at time $t$ \footnote{Note that while here the transition kernel stays the same, we still need a new random variable for each agent at each time. The difference to the policy is that these random variables all parameterize the same function, where as in the policy case, for each time $t$, different functions are parameterized.}.

For any rooted graph $(\mathcal G, \rho)$ we define the $\sigma$-algebra
$$\mathcal F_k ^\rho := \sigma \left( \left\{ \mathbf x_0 ^v, U_t ^v, W_t ^v : \Vert v - \rho \Vert \le k, t < k\right\} \right).$$
The definition of $\mathcal F_k$ tracks all probabilistic inputs that are needed in order to determine the state of the root node at time $k$ (and hence its action at time $k$).

\begin{lemma}
    \label{lemma:action_is_local}
    For any $i \in V(\mathcal G)$, $\mathbf x^i _t$ is measurable with respect to $\mathcal F_{t} ^i$.
\end{lemma}

\begin{proof}
    The statement is verified by induction. For $t=0$ and arbitrary $i \in V(\mathcal G),$ $\mathbf x_0 ^i$ is measurable with respect to $\mathcal F_{0} ^i$ by definition.
    For the induction step, assume $\mathbf x_{t}^v$ is measurable with respect to $\mathcal F^{v} _{t}$ for every node $v$ in the graph.
    By locality of the dynamics, $\mathbf x_{t+1}^i$ is a measurable function of the $1$--hop neighborhood states $\{\mathbf x_t^v : v \in (\mathcal G, i)_1\}$ together with the random variables $U_t^i$ and $W_t^i$.
    Since for each $v \in (\mathcal G,i)_1$ we have $\Vert v-i\Vert \le 1$, the sigma-algebra generated by $\bigcup_{v\in(\mathcal G,i)_1}\mathcal F_t^v$ is contained in $\mathcal F_{t+1}^i$.
    Moreover, $U_t^i$ and $W_t^i$ are contained in $\mathcal F_{t+1}^i$ by definition. Hence $\mathbf x_{t+1}^i$ is measurable with respect to $\mathcal F_{t+1}^i$, which closes the induction.
\end{proof}

We allow the underlying graph structure to be chosen randomly. In order to obtain full information on the graph distribution, we not only need information about the current states of the agents, but also about the graph structure. We therefore introduce the $\sigma$-algebra
$$\mathcal K ^ \rho _ k := \sigma \left( \mathcal F_k ^\rho , \mathcal J \right),$$
where $\mathcal J$ is the canonical $\sigma$-algebra on the space of graphs with at most countably many nodes.

\begin{proof}[Proof of Theorem \ref{thm:locality}]
    Recall the definition of $V_t$ from \eqref{equ:value_fct_def}. Fix $t\in\{0,\dots,T-1\}$ and set $k:=T-t$.
    By Lemma \ref{lemma:action_is_local}, the random variables $( (\mathcal G,\rho),\mathbf x_s )_1$ for $s=t,t+1,\dots,T$ are measurable with respect to $\mathcal K _k ^\rho$ (since information propagates at unit speed and $k=T-t$).
    Since instantaneous rewards depend only on the $1$--hop neighborhood, the reward-to-go from time $t$ to $T-1$ (and hence the value $V_t$) is $\mathcal K_k^\rho$-measurable. This shows that $V_t$ depends only on the $k$-neighborhood information, i.e.\ on the local neighborhood.
\end{proof}

\section{Proof of Theorem \ref{thm:dpp}}
\label{app:proof_dpp}
Recall that the Bellman equation simply states that the value of a state $\boldsymbol \mu$ is simply given as the maximum of the sum of the immediate best reward and the best possible future value.
In order to prove \eqref{equ:dpp} we follow the ideas from \cite{gu2023dynamic}. For any action kernel $h$ and sequence of controls $\pi = (\pi_t)_{t \le T}$, define the IQ function
$$
Q_t (\boldsymbol{\mu},h) := \sup\limits_{\pi \in \Pi_1} Q_t^\pi (\boldsymbol \mu,h).
$$
where
$$
Q_t^\pi (\boldsymbol \mu,h) := \mathbb E ^\pi \left[ \sum\limits_{\ell=t} ^{T-1} \gamma^{\ell-t} r(  ((\mathcal G,\rho), \mathbf x_\ell )_1, u_\ell ) \right],
$$
with $0 < \gamma \le 1$. We abbreviate $Q (\boldsymbol{\mu},h) = Q_0(\boldsymbol{\mu},h)$.
The expectation on the rhs is taken according to $((\mathcal G,\rho), \mathbf x_t)_1 \sim \boldsymbol \mu$ and $u_t \sim h$.
\begin{proof}[Proof of Theorem \ref{thm:dpp}]
    In order to simplify notation we will abbreviate a state of the $1$--hop neighborhood of $\rho$ by $s$. Similarly, we  abbreviate $s_t = ((\mathcal G,\rho), \mathbf x_t)_1$. Define the (average) immediate reward as
    $$
    \hat{r}(\boldsymbol \mu,h) :=  \int \boldsymbol \mu(\mathrm d s) \int h(\mathrm d u \mid s) r(s,u).
    $$
    By definition we have with $\pi^{1:} := (\pi_{t})_{1 \le t \le T}$ that
    $$
    V_0 ^{\pi} (\boldsymbol \mu) = \hat{r}(\boldsymbol \mu,\pi_0) +  \gamma V^{\pi^{1:}} _{1}(T_1\boldsymbol \mu)
    $$
    holds, and similarly for $t>0$. Recall that
    $$
    V (\boldsymbol \mu) = \sup_{\pi \in \Pi_1} V_0 ^{\pi}(\boldsymbol \mu)
    $$
    is defined to be the value of state $\boldsymbol \mu$. We claim
    \begin{equation}
        \label{equ:state_value_identity}
        V_t( \boldsymbol \mu) = \sup\limits_{h \in \mathcal H} Q_t(\boldsymbol \mu,h)
    \end{equation}
    for any state distribution $\boldsymbol \mu$. It suffices to show the claim for $t=0$.
    To show ``$\le$``, note that by definition
    $$
    V_0 ^{\pi} (\boldsymbol \mu) = Q_0^{\pi^{1:}}(\boldsymbol \mu,\pi_0 ),
    $$
    which immediately shows this inequality by taking the sup over all policies and next--step policies and then taking the sup over the policies on the left--hand side.
    For the other direction, fix any start distribution $\pi_0$ and choose $\pi^\epsilon$ such that
    $$
    Q_0^{\pi^\epsilon}(\boldsymbol \mu,\pi_0) \ge Q_0(\boldsymbol \mu,\pi_0)-\epsilon.
    $$
    The existence of $\pi^\epsilon$ is guaranteed by the definition of a supremum. Defining the ``glued policy''
    $$
    \Tilde{ \pi} := (\pi_0,\pi^\epsilon_1,...,\pi^\epsilon _{T-1})
    $$
    and inserting yields the inequality
    $$
    V(\boldsymbol \mu) \ge V_0 ^{\Tilde \pi}(\boldsymbol \mu) = Q_0^{\pi^\epsilon}(\boldsymbol \mu,\pi_0) \ge Q_0(\boldsymbol \mu,\pi_0)-\epsilon.
    $$
    Since $\epsilon > 0$ was arbitrary, equation \eqref{equ:state_value_identity} is shown. The DPP follows now from the calculation
    \begin{equation}
        \nonumber
        \begin{split}
            Q(\boldsymbol \mu,h) &= \sup\limits_{\pi \in \Pi_1} \hat{r}(\boldsymbol \mu,h) + \gamma V^{\pi} _{1}( T_1 \boldsymbol \mu)\\
            &= \hat{r}(\boldsymbol \mu,h) + \gamma \sup\limits_{h' \in \mathcal H} Q_1( T_1 \boldsymbol \mu,h')
        \end{split}
    \end{equation}
and writing out definitions. 
\end{proof}

\section{Proof of Theorem \ref{thm:policy-gradient}}
\label{app:proof_pg}

We provide a detailed proof for the Policy Gradient Theorem applied to the Sparse Mean-Field Control system.

The problem is a Markov Decision Process (MDP) defined by the tuple $(\mathcal{S}, \Psi, \mathcal{T}, R)$, where:
\begin{itemize}
    \item State space $\mathcal{S} = \mathcal{P}(\mathcal{G}_*^{\mathcal{X}})$ (measures over decorated neighborhoods).
    \item Action space $\Psi$ (parameters for the local policy).
    \item Dynamics $\boldsymbol{\mu}_{t+1} = \mathcal{T}(\boldsymbol{\mu}_t, \psi_t)$. Note that given $\boldsymbol{\mu}$ and $\psi$, the next state is deterministic.
    \item Reward $R(\boldsymbol{\mu}, \psi)$.
\end{itemize}
The agent (meta-policy) $\hat{\pi}_\theta(\psi \mid \boldsymbol{\mu})$ is stochastic.
Let $V_t(\boldsymbol{\mu})$ be the value function at time $t$ under policy $\hat{\pi}_\theta$, as
\[ V_t(\boldsymbol{\mu}) = \int_{\Psi} \hat{\pi}_\theta(\psi \mid \boldsymbol{\mu}) \left[ R(\boldsymbol{\mu}, \psi) + V_{t+1}(\mathcal{T}(\boldsymbol{\mu}, \psi)) \right] \mathrm{d}\psi. \]
We aim to find $\nabla_\theta V_0(\boldsymbol{\mu}_0)$. Differentiating $V_t$ with respect to $\theta$ gives
\begin{align}
    \nabla_\theta V_t(\boldsymbol{\mu}) &= \nabla_\theta \int_{\Psi} \hat{\pi}_\theta(\psi \mid \boldsymbol{\mu}) Q_t(\boldsymbol{\mu}, \psi) \mathrm{d}\psi \nonumber \\
    &= \int_{\Psi} \Big[ \nabla_\theta \hat{\pi}_\theta(\psi \mid \boldsymbol{\mu}) Q_t(\boldsymbol{\mu}, \psi) \nonumber + \hat{\pi}_\theta(\psi \mid \boldsymbol{\mu}) \nabla_\theta Q_t(\boldsymbol{\mu}, \psi) \Big] \mathrm{d}\psi. \label{eq:pg_expansion}
\end{align}
Note that this differentiation step is standard in the Policy Gradient literature (e.g., Sutton et al., 1999).
Using the log-derivative trick $\nabla \hat{\pi} = \hat{\pi} \nabla \log \hat{\pi}$,
\begin{align}
    \text{Term 1} &= \mathbb{E}_{\psi \sim \hat{\pi}_\theta} \left[ \nabla_\theta \log \hat{\pi}_\theta(\psi \mid \boldsymbol{\mu}) Q_t(\boldsymbol{\mu}, \psi) \right].
\end{align}

Now consider \textbf{Term 2}, which accounts for the dependency of the future value on $\theta$. Since $R$ and $\mathcal{T}$ do not depend directly on $\theta$ (only via $\psi$), we have
\begin{align}
    \nabla_\theta Q_t(\boldsymbol{\mu}, \psi) &= \nabla_\theta \left[ R(\boldsymbol{\mu}, \psi) + V_{t+1}(\mathcal{T}(\boldsymbol{\mu}, \psi)) \right] \nonumber \\
    &= \nabla_\theta V_{t+1}(\boldsymbol{\mu}_{t+1}) \quad \text{where } \boldsymbol{\mu}_{t+1} = \mathcal{T}(\boldsymbol{\mu}, \psi).
\end{align}
Substituting back into Eq. \eqref{eq:pg_expansion}, we obtain
\begin{align}
    \nabla_\theta V_t(\boldsymbol{\mu}_t) &= \mathbb{E}_{\psi_t} \left[ \nabla_\theta \log \hat{\pi}_\theta(\psi_t \mid \boldsymbol{\mu}_t) Q_t(\boldsymbol{\mu}_t, \psi_t) \right] \nonumber \\
    &\quad + \mathbb{E}_{\psi_t} \left[ \nabla_\theta V_{t+1}(\boldsymbol{\mu}_{t+1}) \right].
\end{align}
This is a recursive relationship. The differentiation of the dynamics is implicitly contained in the term $\nabla_\theta V_{t+1}$. By unrolling this recurrence from $t=0$ to $T-1$, we have
\begin{align}
    \nabla_\theta J(\theta) &= \nabla_\theta V_0(\boldsymbol{\mu}_0) \nonumber \\
    &= \mathbb{E}_{\psi_0} [ \nabla \log \hat{\pi}_0 \cdot Q_0 + \nabla V_1 ] \nonumber \\
    &= \mathbb{E}_{\psi_0} [ \nabla \log \hat{\pi}_0 \cdot Q_0 + \mathbb{E}_{\psi_1} [ \nabla \log \hat{\pi}_1 \cdot Q_1 + \nabla V_2 ] ] \nonumber \\
    &= \sum_{t=0}^{T-1} \mathbb{E}_{\tau} \left[ \nabla_\theta \log \hat{\pi}_\theta(\psi_t \mid \boldsymbol{\mu}_t) Q_t(\boldsymbol{\mu}_t, \psi_t) \right].
\end{align}
The term $\nabla V_T$ vanishes because $V_T(\cdot) = 0$.

\subsection{Proof of Theorem \ref{thm:ctde}}
\label{app:ctde}

\begin{proof}
Let $\boldsymbol{\mu}^N_t$ be the empirical measure state on graph $G_N$, and $\boldsymbol{\mu}_t$ be the limiting measure.
Let the gradient on the finite system be
\[ g^N(\theta) = \sum_{t=0}^{T-1} \mathbb{E} \left[ Q^N_t(\boldsymbol{\mu}^N_t, \psi_t) \nabla_\theta \log \hat{\pi}_\theta(\psi_t \mid \boldsymbol{\mu}^N_t) \right]. \]
Let the gradient on the limit system be
\[ g^\infty(\theta) = \sum_{t=0}^{T-1} \mathbb{E} \left[ Q^\infty_t(\boldsymbol{\mu}_t, \psi_t) \nabla_\theta \log \hat{\pi}_\theta(\psi_t \mid \boldsymbol{\mu}_t) \right]. \]
We assume $\hat{\pi}_\theta$ is Lipschitz continuous and bounded, and $R$ is bounded.

We aim to show $\|g^N - g^\infty\| \to 0$. We decompose the error into three parts via the triangle inequality:

First, under the assumption that the transition dynamics $T$ are continuous with respect to the weak topology and the graph sequence $G_N$ converges in the Benjamini-Schramm sense, we have $\boldsymbol{\mu}^N_t \to \boldsymbol{\mu}_t$ weakly in probability. This follows from the continuous mapping theorem applied to the recursive dynamics.
    
Second, since the reward $r$ is bounded and continuous on decorated neighborhoods, the value function $Q^N_t$ converges to $Q^\infty_t$. Specifically, for any compact set of policies, $Q^N_t(\mu, \psi) \to Q^\infty_t(\mu, \psi)$ uniformly.
    
Third, by Assumption \ref{ass:policy-domination}, $\nabla \log \hat{\pi}_\theta(\psi \mid \mu)$ is continuous in $\mu$.

Now, consider the difference for a fixed time step $t$,
\begin{align*}
    \Delta_t &= \left\| \mathbb{E} [ Q^N_t(\boldsymbol{\mu}^N_t) \nabla \log \pi(\boldsymbol{\mu}^N_t) ] - \mathbb{E} [ Q^\infty_t(\boldsymbol{\mu}_t) \nabla \log \pi(\boldsymbol{\mu}_t) ] \right\| \\
    &\le \underbrace{\left\| \mathbb{E} [ (Q^N_t(\boldsymbol{\mu}^N_t) - Q^\infty_t(\boldsymbol{\mu}^N_t)) \nabla \log \pi(\boldsymbol{\mu}^N_t) ] \right\|}_{\text{(I)}} + \underbrace{\left\| \mathbb{E} [ Q^\infty_t(\boldsymbol{\mu}^N_t) \nabla \log \pi(\boldsymbol{\mu}^N_t) - Q^\infty_t(\boldsymbol{\mu}_t) \nabla \log \pi(\boldsymbol{\mu}_t) ] \right\|}_{\text{(II)}}.
\end{align*}
Term (I): Since $Q^N \to Q^\infty$ uniformly and $\nabla \log \pi$ is bounded, this term vanishes as $N \to \infty$.

Term (II): Let $f(\mu, \psi) = Q^\infty_t(\mu, \psi) \nabla \log \hat{\pi}(\psi \mid \mu)$. This function is bounded and continuous. Since $\boldsymbol{\mu}^N_t \to \boldsymbol{\mu}_t$ weakly, by the definition of weak convergence of measures (and the Portmanteau theorem), $\mathbb{E}[f(\boldsymbol{\mu}^N_t)] \to \mathbb{E}[f(\boldsymbol{\mu}_t)]$.

Thus, $\lim_{N \to \infty} \|g^N - g^\infty\| = 0$. Optimizing the finite graph objective asymptotically optimizes the mean-field objective.
\end{proof}

\section{Proof of Theorem \ref{thm:truncation}}

In the following we provide a proof for Theorem \ref{thm:truncation}. In order to do this, some assumptions are necessary.
\begin{assumption} 
\label{ass:error_propagation_gnn}
    \leavevmode\par
    \begin{itemize}
        \item There exists $\Delta \in \mathbb N$ such that  $\deg (\rho) \le \Delta$ for all $\rho \in \mathcal G$;
        \item  Fix any vertex $\rho  \in \mathcal G$. Let $\mathbf x$ and $\mathbf x'$ be two graph decorations that differ at most by on one node in the $1$-hop neighborhood around $\rho$. I.e.,
        $$ |\{v \in (\mathcal G,\rho)_1  : \mathbf{x}^v \neq (\mathbf{x}')^v \}| \le 1;$$
        Then uniformly in $u \in \mathcal U$,
        $$\left\Vert P \big(\cdot | ( (\mathcal G,\rho), \mathbf{x})_1 ,u \big) - P \big(\cdot | ( (\mathcal G,\rho), \mathbf{x}')_1 ,u \big)  \right\Vert_{\mathrm{TV}} \le \alpha$$
        for $\alpha <1$;
        \item Assume $\alpha\Delta^2 < 1$.
    \end{itemize}
\end{assumption}
In order to avoid confusion with the mean-field value function $V_t(\boldsymbol{\mu})$, we denote the value function for a specific local state $S$ by (the slightly untypical) $W_t(S)$.
\begin{definition}[Local Value Function]
    Let $S = ((\mathcal{G}, \rho), \mathbf{x}) \in \mathcal{G}_*^{\mathcal{X}}$ be a decorated rooted graph state at time $t$. For a policy $\pi$, the local value function $W_t^\pi(S)$ is the expected cumulative reward collected by the agent at the root $\rho$:
    \[
    W_t^\pi(S) := \mathbb{E} \left[ \sum_{s=t}^{T-1} r\big( ((\mathcal{G}, \rho), \mathbf{x}_s)_1, u_s^\rho \big) \;\bigg|\; ((\mathcal{G}, \rho), \mathbf{x}_t) = S \right].
    \]
    The expectation is taken over the action sequence $u_s^\rho \sim \pi_s$ and the state dynamics $\mathbf{x}_{s+1} \sim P$, holding the graph topology $(\mathcal{G}, \rho)$ fixed. Let $W_t^*(S) = \sup_\pi W_t^\pi(S)$ be the optimal local value. We define similarly the local $Q$-function and the optimal (local) action value $Q^*$.
\end{definition}
Note that in the previous display we set the discount factor $\gamma=1$ for simplicity. Next, given a fixed rooted graph $(\mathcal G, \rho)$, we define the weighted distance by
\[
d_{\lambda}(S, S') := \sum_{v \in V(G): \mathbf{x}(v) \neq \mathbf{x}'(v)} \lambda^{\Vert \rho - v \Vert}.
\]
\begin{proposition}
\label{prop:value_stability}
Let $(\mathcal G, \rho)$ be a fixed rooted graph. Let $S = ((\mathcal G, \rho), \mathbf{x})$ and $S' = ((\mathcal G, \rho), \mathbf{x}')$ be two states with possibly different node decorations $\mathbf{x}, \mathbf{x}'$.
Let $\lambda := \alpha \Delta < 1$ by Assumption \ref{ass:error_propagation_gnn}. 
Then, the optimal local value function $W_t^*(S)$ satisfies:
\[
|W_t^*(S) - W_t^*(S')| \le L_t \, d_{\lambda}(S, S'),
\]
where $L_t < \infty$ is a constant depending on $T$ but is uniform in $(\mathcal G, \rho)$.
\end{proposition}

\begin{proof}
We proceed by backward induction on $t$. For $t=T$,
 the inequality holds trivially.
Next, assume the bound holds for the function $W_{t+1}^*$. That is, we assume $W_{t+1}^*$ is Lipschitz continuous with constant $L_{t+1}$ 
The Bellman equation for $W_t^*$ reads
\[
W_t^*(S) = \sup_{u \in \mathcal{U}} \left\{ r((S)_1, u) + \mathbb{E}[W_{t+1}^*(S_{\mathrm{next}}) \mid S, u] \right\}.
\]
Fix states $S, S'$. Let $D = \{ v \in V(\mathcal G) : \mathbf{x}^v \neq (\mathbf{x}')^v \}$ be the (deterministic) set of disagreeing nodes. Using $|\sup f - \sup g| \le \sup |f - g|$, we have the estimate
\begin{align*}
    |W_t^*(S) - W_t^*(S')| &\le \sup_{u \in \mathcal{U}} \bigg( \underbrace{|r((S)_1, u) - r((S')_1, u)|}_{\textbf{(I)}} \\
    &\quad + \underbrace{\left| \mathbb{E}[W_{t+1}^*(S_{\mathrm{next}}) \mid S, u] - \mathbb{E}[W_{t+1}^*(S'_{\mathrm{next}}) \mid S', u] \right|}_{\textbf{(II)}} \bigg).
\end{align*}
We now bound both terms separately, starting with $\mathbf{(I)}$. Since $r$ is local (radius 1) and bounded, $|r(S) - r(S')| \le C_r \sum_{v \in D \cap B(\rho, 1)} 1$. In particular, $1 = \lambda^{-1} \lambda^{\Vert \rho - v \Vert}$ for $\rho,v$ being neighbors. Thus 
$$\text{(I)} \le (C_r/\lambda) d_{\lambda}(S, S').$$

For $\mathbf{(II)}$, we apply the induction hypothesis to every realization, which yields
\[
\text{(II)} \le \mathbb{E} \left[ |W_{t+1}^*(S_{\mathrm{next}}) - W_{t+1}^*(S'_{\mathrm{next}})| \right] \le L_{t+1} \mathbb{E} \left[ d_{\lambda}(S_{\mathrm{\mathrm{next}}}, S'_{\mathrm{next}}) \right].
\]
Now, expand the expectation of the distance. Let $\mathbf{x}_{\mathrm{next}}$ and $\mathbf{x}'_{\mathrm{next}}$ be the random decorations at time $t+1$. By Assumption \ref{ass:error_propagation_gnn}, for any node $w$ it follows from the triangle inequality that
\[
\mathbb{P}\left(\mathbf{x}_{\mathrm{next}}^w \neq (\mathbf{x}'_{\mathrm{next}})^w \right) \le \alpha \sum_{v \sim w} \mathbf{1}_{v \in D}.
\]
By substituting this probability bound into the expectation one obtains
\begin{align*}
    \mathbb{E}[d_{\lambda}] &= \sum_{w \in V(\mathcal G)} \lambda^{\Vert \rho - w \Vert} \mathbb{P}(\mathbf{x}_{\mathrm{next}}^w \neq \left(\mathbf{x}'_{\mathrm{next}})^w \right) \\
    &\le \sum_{w} \lambda^{\Vert \rho - w \Vert} \left( \alpha \sum_{v \sim w} \mathbf{1}_{v \in D} \right).
\end{align*}
Swapping sums thus gives the estimate
\[
\mathbb{E}[d_{\lambda}]\le \alpha \sum_{v \in D} \sum_{w \sim v} \lambda^{\Vert \rho - w \Vert}.
\]
Using the triangle inequality $\Vert \rho - w \Vert \ge \Vert \rho - v \Vert - 1$, we have (as $\lambda <1$) $\lambda^{\Vert \rho - w \Vert} \le \lambda^{-1} \lambda^{\Vert \rho - v \Vert}$. With max degree $\Delta$ this means that
\[
\mathbb{E}[d_{\lambda}]\le \alpha \Delta \lambda^{-1} \sum_{v \in D} \lambda^{\Vert \rho - v \Vert}.
\]
Since $\lambda = \alpha \Delta$, we have $\alpha \Delta \lambda^{-1} = 1$. Thus $\text{(II)} \le L_{t+1} d_{\lambda}(S, S')$.
Combining terms gives the recursion $L_t = C_r/\lambda + L_{t+1}$. Since the horizon $T$ is finite, $L_t < \infty$ and the claim is shown.
\end{proof}
\begin{lemma}
\label{lem:perf_diff}
    For any two Markov kernels $\pi, \tilde{\pi}$ and any finite horizon $T \in \mathbb N$ we have
    \[
    J(\pi) - J(\tilde{\pi}) = \sum_{t=0}^{T-1} \mathbb{E}_{\tau \sim \tilde{\pi}} \left[ \mathbb{E}_{u \sim \pi(\cdot|S_t)}[Q_t^\pi(S_t, u)] - \mathbb{E}_{\tilde{u} \sim \tilde{\pi}(\cdot|S_t)}[Q_t^\pi(S_t, \tilde{u})] \right].
    \]
\end{lemma}
\begin{proof}
    Recall that $W_t^\pi(S) = \mathbb{E}_{u \sim \pi(\cdot|S)}[Q_t^\pi(S, u)]$ is the value function of state $S$ with policy $\pi$. 
    By definition, the Q-function satisfies the Bellman recursion
    \[
    Q_t^\pi(S, \tilde{u}) = r(S, \tilde{u}) + \mathbb{E}_{S' \sim P(\cdot|S, \tilde{u})}[V_{t+1}^\pi(S')],
    \]
    with $W_T^\pi(S) = 0$.
    Consider the term inside the summation on the RHS. Let $\tau \sim \tilde{\pi}$ denote a trajectory generated by $\tilde{\pi}$. At time $t$ we define the advantage
    \begin{align*}
        A_t &:= \mathbb{E}_{u \sim \pi}[Q_t^\pi(S_t, u)] - \mathbb{E}_{\tilde{u} \sim \tilde{\pi}}[Q_t^\pi(S_t, \tilde{u})] \\
            &= W_t^\pi(S_t) - \mathbb{E}_{\tilde{u} \sim \tilde{\pi}}\left[ r(S_t, \tilde{u}) + \mathbb{E}_{S_{t+1} \sim P(\cdot \mid S_t, \Tilde u)}[W_{t+1}^\pi(S_{t+1})] \right].
    \end{align*}
    $A_t$ essentially measures the ``loss`` one gets when running $\Tilde{\pi}$ instead of $\pi$ for a single time step.
    Sample $S_t$ according to the dynamics with respect to $\Tilde{\pi}$, which yields
    \[
    \mathbb{E}_{\tau \sim \tilde{\pi}}[A_t] = \mathbb{E}_{\tau \sim \tilde{\pi}}[W_t^\pi(S_t)] - \mathbb{E}_{\tau \sim \tilde{\pi}}[r(S_t, \tilde{u})] - \mathbb{E}_{\tau \sim \tilde{\pi}}[W_{t+1}^\pi(S_{t+1})].
    \]
    Summing this over $t=0$ to $T-1$, one therefore obtains
    \[
    \sum_{t=0}^{T-1} \mathbb{E}_{\tau \sim \tilde{\pi}}[A_t] = \sum_{t=0}^{T-1} \left( \mathbb{E}_{\tau \sim \tilde{\pi}}[W_t^\pi(S_t) - W_{t+1}^\pi(S_{t+1})] \right) - \underbrace{\sum_{t=0}^{T-1} \mathbb{E}_{\tau \sim \tilde{\pi}}[r(S_t, \tilde{u})]}_{J(\tilde{\pi})}.
    \]
    Note that the first sum telescopes.
    Since $W_T^\pi \equiv 0$ and $\mathbb{E}[W_0^\pi(S_0)] = J(\pi)$, we obtain
    \[
    \sum_{t=0}^{T-1} \mathbb{E}_{\tau \sim \tilde{\pi}}[A_t] = J(\pi) - J(\tilde{\pi}). \qedhere
    \]
\end{proof}

\begin{proof}[Proof of Theorem \ref{thm:truncation}]
    Let $\pi^*$ be the optimal policy.
    Let $\mathcal T _k (S)$ be the state that has every node decoration at distance $\Vert \rho - v \Vert > k$  set to some fixed state, say $0$. In the $k$-ball around $\rho$ we take $S = \mathcal T _K (S)$. Next, define the truncated policy $\tilde{\pi}$ by $\tilde{\pi}(\cdot|S) := \pi^*(\cdot|\mathcal{T}_k(S))$.
    Since $\pi^{(k)}$ denotes the optimal $k$-local policy, we have $J(\pi^*) - J(\pi^{(k)}) \le J(\pi^*) - J(\tilde{\pi})$.
    
    Our plan is to apply Lemma \ref{lem:perf_diff}. Let $S_t$ be the state at time $t$ and $S'_t = \mathcal{T}_k(S_t)$ be its truncation. The loss term at time $t$ is
    \[
    \Delta_t = \mathbb{E}_{u \sim \pi^*(\cdot|S_t)}[Q_t^*(S_t, u)] - \mathbb{E}_{\tilde{u} \sim \pi^*(\cdot|S'_t)}[Q_t^*(S_t, \tilde{u})],
    \]
    because $\pi^* (\cdot \mid S_t ') = \Tilde{\pi}(\cdot \mid S_t)$.
    Note that the first term is simply the optimal value function $V_t^*(S_t)$. The second term involves actions $\tilde{u}$ chosen according to the truncated state $S'_t$, but evaluated on the real state $S_t$.
    
    We add and subtract the optimal value of the truncated state, $W_t^*(S'_t) = \mathbb{E}_{\tilde{u} \sim \pi^*(\cdot|S'_t)}[Q_t^*(S'_t, \tilde{u})]$ so that we obtain
    \begin{align*}
        \Delta_t &= W_t^*(S_t) - W_t^*(S'_t) + W_t^*(S'_t) - \mathbb{E}_{\tilde{u} \sim \pi^*(\cdot|S'_t)}[Q_t^*(S_t, \tilde{u})] \\
        &= \underbrace{(W_t^*(S_t) - W_t^*(S'_t))}_{\text{(I)}} + \underbrace{\mathbb{E}_{\tilde{u} \sim \pi^*(\cdot|S'_t)} [Q_t^*(S'_t, \tilde{u}) - Q_t^*(S_t, \tilde{u})]}_{\text{(II)}}.
    \end{align*}
    
    We bound both terms separately.
    For the first term (I), observe that by the Lipschitz continuity of the optimal value function, we have
        \[
        |W_t^*(S_t) - W_t^*(S'_t)| \le L_Q d_\lambda(S_t, S'_t).
        \]
        For the second term (II), note that this is the expectation of the difference in Q-values for the \textit{same} action $\tilde{u}$ but different states. By the Lipschitz continuity of $Q^*$ (which follows from the Lipschitz estimate on $W^* _t$)
        \[
        |Q_t^*(S'_t, \tilde{u}) - Q_t^*(S_t, \tilde{u})| \le L_Q d_\lambda(S_t, S'_t).
        \]
        Taking the expectation over $\tilde{u}$ does not change this bound.
    Combining these estimates, we get $\Delta_t \le 2 L_Q d_\lambda(S_t, \mathcal{T}_k(S_t))$. Finally, we bound the distance. Since $S_t$ and $\mathcal{T}_k(S_t)$ agree on the $k$-hop neighborhood, disagreements occur only at distance $j > k$.
    \[
    d_{\lambda}(S_t, \mathcal{T}_k(S_t)) \le \sum_{j=k+1}^\infty \Delta^j \lambda^j = \frac{(\Delta \lambda)^{k+1}}{1 - \Delta \lambda}.
    \]
    Summing over $t=0 \dots T-1$ yields the desired bound
    \[
    J(\pi^*) - J(\pi^{(k)}) \le \sum_{t=0}^{T-1} 2 L_Q \frac{(\Delta \lambda)^{k+1}}{1 - \Delta \lambda} = O((\Delta \lambda)^k). \qedhere
    \]
\end{proof}

\section{Experimental Setup}
\label{app:exp}
\subsection{SIR Model Dynamics}
The experiments are conducted on a graph $\mathcal G=(V, E)$ where each node represents an agent. The system follows a discrete-time \textbf{SIR (Susceptible-Infected-Recovered)} process with interventions.
\begin{itemize}
    \item \textbf{States:} each node $i$ is in one of three states at time $t$: $S_i^t \in \{S, I, R\}$.
    \begin{itemize}
        \item $S$ (Susceptible): can be infected.
        \item $I$ (Infected): can infect neighbors.
        \item $R$ (Recovered/Removed): immune, cannot infect or be infected (includes vaccinated nodes).
    \end{itemize}
    \item \textbf{Infection Dynamics:} a susceptible node $i$ becomes infected at time $t+1$ with probability $P_{\mathrm{inf}}$. If node $i$ decides to \textit{isolate}, the susceptibility is reduced by a factor $f_{\mathrm{iso}} \in [0, 1]$,
    \[ P(S_i^{t+1} = I | S_i^t = S) = 1 - \prod_{j \in \mathcal{N}(i): S_j^t = I} \left( 1 - \beta \cdot (1 - 1_{a_i=2} \cdot f_{\mathrm{iso}}) \right) \]
    where $\beta$ is the base transmission probability per edge.
    \item \textbf{Recovery Dynamics:} an infected node recovers naturally with probability $\gamma$ at each step, meaning
    \[ P(S_i^{t+1} = R | S_i^t = I) = \gamma \]
    \item \textbf{Vaccination:} an explicit action ($a_i=1$) transitions a susceptible node directly to the recovered state ($S \to R$), granting immediate immunity.
\end{itemize}
The reward for every agent $i$ at time $t$ reads 
 \[
r_i^t = - \left[ 
    c_I \cdot 1(S_i^t = I) 
    + c_V \cdot 1(S_i^t \in \mathcal{V}) 
    + c_Q \cdot 1(a_i^t = \text{Isolate}) 
    + C_{\mathrm{global}} \cdot 1\left( \frac{1}{N} \sum_{j=1}^{N} 1(S_j^t = I) > \tau \right)
\right],
\]
\noindent where
\begin{itemize}
    \item $c_I, c_V, c_Q$ are the costs for infection, vaccination, and isolation respectively;
    \item $1(\cdot)$ is the indicator function;
    \item $\mathcal{V}$ denotes the set of vaccinated states (distinct from naturally recovered);
    \item $C_{\mathrm{global}}$ is the penalty applied when the global infection rate exceeds the threshold $\tau$.
\end{itemize}

\subsection{Setup}
Both our experiments utilize the architecture that is described in Figure \ref{fig:architecture}. A GNN (parameterized by $\phi$) predicts the means $\mu$ of a policy parameter distribution based on the global state $S_{\mathrm{global}}$. The actual policy parameters $\theta$ for an episode are sampled from a Gaussian distribution centered at $\mu = \mathrm{GNN}(S_{\mathrm{global}})$ with a fixed standard deviation ($\sigma=3.0$), meaning
\[ \theta \sim \mathcal{N}(\mathrm{GNN}(S_{\mathrm{global}}); \phi), \sigma^2 I). \]
The local agents operate using a decentralized policy parameterized by $\theta$. The policy then maps a discrete local observation state $s_{\mathrm{local}}$ to action probabilities. The state space is characterized as follows.
The local state space is defined by the agent's own state $s_i \in \{S, I, R\}$ and binary flags indicating the presence of neighbors in each state. To decrease the dimensions of the learning problem, the policy is restricted to detect only the presence of \textit{infected} neighbors, having no information about $S$ and $R$ neighbors. This effectively reduces the decision space to $3 \times 2 = 6$ functional states: $\{S, I, R\} \times \{0,1\}$. All experiments were run on a Macbook Pro ($2024$) and across $10$ random seeds. We use NetworkX \cite{hagberg2008networkx} to generate our graphs and the Adam optimizer \cite{kingma2014adam} to optimize the learnable parameters. Everything is implemented in PyTorch \cite{paszke2019pytorch}.

\subsection*{Experiment 1: Graph-Lifted State Space}
\noindent This setup focuses on controlling an outbreak on a lattice using isolation/vaccination. 
As graph we chose a $20 \times 20$ grid, meaning $N=400$ nodes.
The environment parameters are chosen so that the infection would spread supercritically in case no controller is present. In particular, we choose the (average) cost of being infected greater than the cost of vaccination ($c_v = 1, c_I = 0.5, c_Q = 0.2$). By choosing the recovery rate to be suitably low ($\gamma=0.2$) and the infection probability suitably high ($\beta = 0.25$), the optimal policy is to immediately stop the infection from the get go and not hope for it to die out on its own. 

We found that using the aggregated mean-field (i.e. $k>0$ message passes for the GNN encoder) does not improve performance for the model. This is to be expected; the optimal policy can be inferred from local knowledge alone.

\subsection*{Experiment 2: Sensitivity to Initial Conditions}
\noindent This setup tests the ability to manage infection in a graph of disjoint cliques under a hard vaccination budget constraint.
The graph topology is now $N=20$ disjoint triangles.
For simplicity we set $\beta=1.0$, $\gamma=0.0$, meaning no natural recovery and guaranteed infection spread. We now start in one of two initial scenarios.
    \begin{itemize}
        \item \textit{Scattered:} 5 infected nodes, each in a distinct triangle.
        \item \textit{Concentrated:} 5 infected nodes concentrated in 4 triangles.
    \end{itemize}
While vaccination is now very expensive ($c_v = 50$) and being in state $I$ is cheap ($c_I = 0$, again choosing $c_Q = 0.2$), a global penalty is applied if more than $25$ percent of all triangles have at least two infected. The meta policy has thus to distinguish the scattered state and the concentrated state: in the scattered state, vaccination must happen immediately for all nodes with infected neighbors; otherwise, the global (severe) punishment is applied. In the concentrated state, however, the global punishment threshold can never be achieved, meaning that it is cheaper to let agents become infected in the local triangles rather than to vaccinate. Here, we find that the GNN can distinguish both initial conditions, where as the mean field can not (note that the mean field is constant across both different initial scenarios). This validates that local neighborhood information on the mean-field level can be important to find optimal strategies.
\end{document}